\documentclass[12pt]{article}
\usepackage{amsmath}
\usepackage{graphicx}
\usepackage{tabularx,supertabular,amsmath,latexsym,amssymb}
\usepackage{float,multirow,rotating,times}
\usepackage{upgreek,wrapfig}
\usepackage{enumerate}
\usepackage[authoryear]{natbib}
\usepackage{comment}

\newcommand{\blind}{0}

\newcommand {\ctn}{\citet}

\renewcommand{\t}{\ensuremath{\theta}}
\renewcommand{\a}{\ensuremath{\alpha}}
\renewcommand{\b}{\ensuremath{\beta}}

\newcommand{\e}{\ensuremath{\epsilon}}
\newcommand{\from}{\ensuremath{\leftarrow}}
\newcommand{\bm}{\mathbf}

\newcommand {\bmu}{\mbox{\boldmath $\mu$}}
\newcommand {\bSigma}{\mbox{\boldmath $\Sigma$}}
\newcommand{\dint}{\int\displaylimits}

\newcommand{\bC}{\mathbf C}
\newcommand{\bp}{\mathbf p}
\newcommand{\bq}{\mathbf q}

\newcommand{\bx}{\mathbf x}
\newcommand{\by}{\mathbf y}
\newcommand{\bz}{\mathbf z}
\newcommand{\bw}{\mathbf w}
\newcommand{\be}{\pmb\e}

\newcommand{\topline}{\hrule height 1pt width \textwidth \vspace*{2pt}}
\newcommand{\botline}{\vspace*{2pt}\hrule height 1pt width \textwidth \vspace*{4pt}}
\newtheorem{algo}{Algorithm} 

\numberwithin{equation}{section}
\numberwithin{algo}{section}
\numberwithin{table}{section}
\numberwithin{figure}{section}

\newcommand{\statesp}{\ensuremath{\mathcal X}}
\newcommand{\Y}{\ensuremath{\mathcal Y}}
\newcommand{\D}{\ensuremath{\mathcal D}}

\newcommand{\R}{\ensuremath{\mathbb R}}

\newcommand{\supr}[2]{{#1}^{(#2)}}

\newcolumntype{Z}{>{\centering\arraybackslash}X}%

\date{}
\begin{document}

\bibliographystyle{natbib}

\def\spacingset#1{\renewcommand{\baselinestretch}%
{#1}\small\normalsize} \spacingset{1}


\if0\blind
{
  \title{\bf Markov Chain Monte Carlo Based on Deterministic Transformations}
   \author{Somak Dutta\thanks{Corresponding e-mail: sdutta@galton.uchicago.edu}\\
    Department of Statistics\\
    University of Chicago\\
    and\\
   Sourabh Bhattacharya\\
   Bayesian and Interdisciplinary Research Unit\\
  Indian Statistical Institute
 }
\maketitle
 } \fi

\if1\blind
{
 \bigskip
 \bigskip
 \bigskip
 \begin{center}
 {\LARGE\bf Markov Chain Monte Carlo Based on Deterministic Transformations}
 \end{center}

  \medskip
  } \fi

  \bigskip

\begin{abstract}
In this article we propose a novel MCMC method based on deterministic 
transformations $T$ : $\statesp$$\times$$\D$ $\to$ $\statesp$ 
where $\statesp$ is the state-space and $\D$ is some set which may or may not be a subset of $\statesp$. 
We refer to our new methodology as
Transformation-based Markov chain Monte Carlo (TMCMC). 
One of the remarkable advantages of our proposal is that even if the underlying target
distribution is very high-dimensional, deterministic transformation of a one-dimensional
random variable is sufficient to generate an appropriate Markov chain that is guaranteed
to converge to the high-dimensional target distribution. Apart from clearly leading to massive
computational savings, this idea of deterministically transforming a single random variable
very generally leads to excellent acceptance rates, even though all the random variables associated
with the high-dimensional target distribution are updated in
a single block. Since it is well-known that joint updating of many random variables using
Metropolis-Hastings (MH) algorithm generally leads to poor acceptance rates, TMCMC, in this regard,
seems to provide a significant advance. 
We validate our proposal theoretically, establishing the convergence properties. 
Furthermore, we show
that TMCMC can be very effectively adopted for simulating from doubly intractable distributions.

We show that TMCMC includes hybrid Monte Carlo (HMC) as a special case. We also contrast TMCMC with
the generalized Gibbs and Metropolis methods of \ctn{Liu99}, \ctn{Liu00} and \ctn{Kou05},
pointing out that even though the latter also use transformations, their goal is to seek improvement 
of the standard Gibbs and Metropolis Hastings
algorithms by adding a transformation-based step, while TMCMC is an altogether new and general 
methodology for simulating from  
intractable, particularly, high-dimensional distributions.

TMCMC is compared with MH using the well-known Challenger data, demonstrating the effectiveness of 
of the former in the case of highly correlated variables. Moreover, we apply our methodology to a 
challenging posterior simulation problem associated with the geostatistical model of \ctn{Diggle98}, 
updating 160 unknown parameters
jointly, using a deterministic transformation of a one-dimensional random variable. Remarkable
computational savings as well as good convergence properties and acceptance rates are the results. 

\end{abstract}
\noindent%

\spacingset{1.45}

\section{Introduction}

Markov chain Monte Carlo (MCMC) has revolutionized statistical, particularly, Bayesian
computation. In the Bayesian paradigm, however complicated the posterior distribution may be, 
it is always possible, in principle, to obtain as many (dependent) samples from the posterior
as desired, to make inferences about posterior characteristics. But in spite of the obvious 
success story enjoyed by the theoretical side of MCMC, satisfactory practical implementation 
of MCMC often encounters severe challenges, particularly in very high-dimensional problems.
These challenges may arise in the form of the requirement of enormous computational effort, often 
requiring inversions of very high-dimensional matrices, implying the requirement
of enormous computation time, even for a single iteration. Given that such high-dimensional
problems typically converge extremely slowly to the target distribution triggered by complicated 
posterior dependence structures between the unknown parameters, astronomically large number
of iterations (of the order of millions) are usually necessary. This, coupled with the computational
expense of individual iterations, generally makes satisfactory implementation of MCMC,
and hence, satisfactory Bayesian inference, infeasible.  That this is the situation despite steady
technological advancement, is somewhat disconcerting.

\subsection{Overview of the contributions of this paper}
\label{subsec:contributions}
In an attempt to overcome the problems mentioned above, 
in this paper we propose a novel methodology that can jointly update all the unknown parameters
without compromising the acceptance rate, unlike in Metropolis-Hastings (MH) algorithm. In fact,
we show that even though a very large number of parameters are to be updated, these can be
updated by simple deterministic transformations of a single, one-dimensional random variable,
the distribution of which can be chosen very flexibly. As can be already anticipated from this brief description,
indeed, this yields an extremely fast simulation algorithm, thanks to the singleton random variable
to be flexibly simulated, and the subsequent simple deterministic transformation, for example,
additive transformation. It is also possible, maybe more efficient sometimes, to generate more
than one, rather than a single, random variables, from a flexible multivariate (generally independent), but 
low-dimensional distribution. We refer to our new methodology as Transformation-based MCMC (TMCMC). 

We show that by generating as many random variables
as the number of parameters, instead of a single/few random variables, TMCMC can be reduced
to a MH algorithm with a specialized proposal distribution. 
Another popular MCMC methodology, the hybrid Monte Carlo (HMC) method, which relies upon
a specialized deterministic transformation, will be shown to be a special case of TMCMC. 

We also provide a brief overview of the transformation-based generalized Gibbs and Metropolis
methods of \ctn{Liu99}, \ctn{Liu00} and \ctn{Kou05}, and point out their differences with TMCMC,
also arguing that TMCMC can be far more efficient at least in terms of computational gains.

Apart from illustrating TMCMC on the well-known Challenger data set, and demonstrating its superiority
over existing MH methods, we successfully apply TMCMC with the mere simulation of a single random variable,
to update 160 unknown parameters in every iteration, in the challenging geospatial problem of
\ctn{Diggle98}. The computational challenges involved with this and similar
geospatial problems have motivated varieties
of MCMC algorithms and deterministic approximations to the posterior in the literature 
(see, {\em e. g.} \ctn{Rue09}, \ctn{Chris06}
and the references therein). With our TMCMC algorithm we have been able to perform $5.5\times 10^7$
iterations (in a few days) and obtain reasonable convergence. 

We also show how TMCMC can be adopted to significantly improve computational efficiency in doubly intractable
problems, where the posterior, apart from being intractable, also involves the normalizing constant
of the likelihood---the crucial point being that the normalizing constant, which depends upon unknown parameters,
is also intractable.

The rest of this article is structured as follows. 
In Section \ref{sec:gtm} we introduce our new TMCMC method based on transformations. The univariate and the 
multivariate cases are considered separately in Sections \ref{sec:gtm:univar} and \ref{sec:gtm:multivar} 
respectively. In Section \ref{sec:single_e} we study in details the role and efficiency of a singleton $\e$ 
in updating high-dimensional Markov chains using TMCMC.
Illustration of TMCMC with singleton $\e$ using the Challenger data and comparison with a popular MCMC technique are provided in Section \ref{sec:appl:logistic}.
Application of TMCMC with single $\e$ to the 160-dimensional geospatial problem of \ctn{Diggle98} is detailed in  Section \ref{sec:appl:dtm}.
Section \ref{sec:doublyintract} shows how TMCMC may be applied to the bridge-exchange algorithm of \cite{Murray06} in doubly intractable problems to speed-up computation. 
Finally, conclusions and overview of future work are provided in Section \ref{sec:conclusions}.

Further investigations and additional details are provided in the supplement \ctn{Dutta13b}, 
whose sections, figures and tables have the prefix ``S-" when referred to in this paper. 
The contents of the supplement are as follows.
Section S-1 contains the proof of detailed balance for TMCMC, Section S-2 provides the general TMCMC algorithm for a one-dimensional proposal,
while Section S-3 contains details on convergence properties of additive TMCMC.
In Section S-4 we provide a more structured version of the general TMCMC algorithm of Section S-2, 
proving detailed balance
of this algorithm in Section S-5.
Detailed investigation of acceptance rate of additive TMCMC and comparison with that of random walk MH (RWMH) is 
carried out in Section S-6. In Sections S-7
and S-8 respectively, comparisons of TMCMC with HMC and generalized Gibbs/Metropolis methods of \ctn{Liu99}, \ctn{Liu00}, and \ctn{Kou05}
are provided. Examples of TMCMC for discrete state spaces are provided
in Section S-9.

\section{MCMC algorithms based on transformations on the state--space} \label{sec:gtm}
In this section we propose and study the TMCMC algorithms. First, we construct it for state-spaces of dimension one. 
This case is not of much interest because the state space is similar to the real line and numerical integration is quite efficient in this scenario. 
Nevertheless, construction of the TMCMC algorithm for one dimensional problems helps to generalize it to higher dimensions and points out its connections 
(similarities in one-dimension and dissimilarities in higher dimensions) with the MH algorithm. In Sections \ref{sec:gtm:multivar} 
and \ref{sec:single_e} the TMCMC algorithm is 
generalized to higher dimensional state-spaces, the latter section considering the utility of single $\e$ in high dimensions.

\subsection{Univariate case}\label{sec:gtm:univar}

Before providing the formal theory we first provide an informal discussion of our ideas with 
a simple example involving the additive transformation.

\subsubsection{Informal discussion}

In order to obtain a valid algorithm based on transformations, we need to design appropriate ``move types"
so that detailed balance and irreducibility hold. Given that we are in the current state $x$, 
we can propose the ``forward move" $x'=x+\e$; here $\e>0$ is a simulation from some arbitrary density of the
form $g(\e)I_{(0,\infty)}(\e)$. To move back to $x$ from $x'$, we need to apply the ``backward transformation"
$x'-\e$. In general, given $\e$ and the current state $x$, we shall denote the forward transformation by 
$T(x,\e)$, and the backward transformation by $T^b(x,\e)$. 

The forward and the backward transformations need to be 1-to-1. In other words, for any fixed $\e$,
given $x'$ the backward transformation must be such that $x$ can be retrieved uniquely.
Since this must hold good for every $x$ in the state space, the transformation must be onto as well.
Similarly, for any fixed $\e$, there must exist $x$ such that the forward 
transformation leads to arbitrarily chosen $x'$ in the state space uniquely, implying that this 
transformation is also 1-to-1 and onto.
If, given $\e$ and $x'$, say, more than one solution exist, then return to the current value $x$   
can not be ensured, and this makes detailed balance, a requirement for stationarity of the 
underlying Markov chain, hard to satisfy.

The detailed balance requirement also demands that, given $x$, the regions covered by the forward and
the backward transformations are disjoint. 
For example, in our additive transformation case, the forward transformation always takes 
$x$ to some unique $x'$, where $x'>x$. To return from $x'$ to $x$, it is imperative that
the backward transformation decreases the value of $x'$ to give back $x$. Thus, if the
forward transformation always increases the current value $x$, the backward transformation
must always decrease $x$. In other words, the regions covered by the two transformations are disjoint.  
Since $x$ is led to $x'$ by the forward transformation and $x'$ is taken back to $x$
by the backward transformation, we must have $T(T^b(x,\e),\e)=x$. Also, the sequence of 
forward and backward transformations can be changed to achieve the same effect, that is, we must also
have $T^b(T(x,\e),\e)=x$.
In the above discussion we indicated the use the same $\e$ for updating $x$ to $x'$ and 
for moving back from $x'$ to $x$. An important advantage associated with this strategy is that
whatever the choice of the density $g(\e)I_{(0,\infty)}(\e)$, it will cancel in the acceptance ratio of our 
TMCMC algorithm, resulting in a welcome simplification.

Thanks to bijection each of the forward and the backward transformations will be equipped with their respective
inverses. In general, we denote by $T(x,\e)$ and $T^b(x,\e)$ the forward and the backward transformations,
and by $T^{-1}(x,\e)$ and ${T^b}^{-1}(x,\e)$ their respective inverses. Note that for fixed $\e$,
$T^{-1}(x,\e)=T^b(x,\e)$, and  ${T^b}^{-1}(x,\e)=T(x,\e)$, but the general inverses
must be defined by eliminating $\e$. For instance, substituting
$\e=x'-x$ for the forward transformation yields $T(x,\e)=T(x,x'-x)=x+(x'-x)=x'$.
Defining $T^{-1}(x,x')=x'-x$, it then follows that $T(x,T^{-1}(x,x'))$$=x'$$=$$T^{-1}(x,T(x,x'))$,
showing that $T^{-1}$ is the inverse of $T$ in the above sense. Similarly, ${T^b}^{-1}$ can also
be defined.

\subsubsection{Formal set-up}

Suppose $T : \statesp \times \mathcal D \to \statesp$ for some $\mathcal D$ (possibly a subset of \statesp) 
is a totally differentiable transformation such that 
\begin{enumerate}
 \item for every fixed $\e \notin \mathcal N_1$, the transform $x \longmapsto T(x,\e)$ is bijective and 
 differentiable and that the inverse is also differentiable.
 \item for every fixed $x \notin \mathcal N_2$, the transform $\e \longmapsto T(x,\e)$ is injective.
\end{enumerate}
where $\mathcal N_1$ and $\mathcal N_2$ are $\pi$-negligible sets. Further suppose that the Jacobian
 \[ J(x,\e) \quad = \quad \left| \dfrac{\partial (T(x,\e),\e)}{\partial(x,\e)} \right|\]
is non-zero almost everywhere.

Suppose there is a subset $\Y$ of $\D$ such that  $\forall x \notin \mathcal N_2$ the 
sets $T(x,\Y)$ and $T^b(x,\Y)$ are disjoint, where $T^b(x,\e)$ is the backward transformation defined by:
\[ T\left( T^b(x,\e),\e\right) = T^b\left( T(x,\e),\e\right) = x\]

\textbf{Example: Transformations on One dimensional state--space}
\begin{enumerate}
 \item(additive transformation) Suppose $\statesp = \D = \mathbb R$ and $T(x,\e) = x + \e$.  
 Let $T^b(x,\e) = x - \e$.  
 This transformation is basically  the random walk if $\e$ is a random quantity. Notice that if we may 
 choose $\Y = (0,\infty)$, then $T(x,\Y) = (x,\infty)$, $T^b(x,\Y) = (-\infty,x)$ and we can characterize 
 the transformation as a forward move or a backward move according as $\e \in$ or $\notin \Y$. Notice that 
 here $\mathcal N$ is the empty set and for all $\e \in \D$ the map $x \longmapsto T(x,\e)$ is a bijection.

\item(log--additive transformation) Suppose $\statesp = \D = (0,\infty)$ and $T(x,\e) = x\e$. 
For all $x \in \statesp$, $T^b(x,\e) = x/\e$; $\Y = (0,1)$.

\item(multiplicative transformation) Let $\statesp = \mathbb R = \D$, $T(x,\e) = x\e$. 
Then $\mathcal N_1 = \mathcal N_2 = \{0\}$, 
for all $\e \ne 0$, $T^b(x,\e) = x/\e$; $\Y = (-1,1) - \{0\}$.

\end{enumerate}

Suppose further that $g$ is a  density on $\Y$ and that $0 < p < 1$. Then the MCMC algorithm based 
on transformation is given in Algorithm \ref{algo:GTM:univar}
\vspace{2cm}
\begin{algo}\label{algo:GTM:univar} \topline MCMC algorithm based on transformation 
(univariate case) \botline \normalfont \ttfamily
\begin{itemize}
 \item Input: Initial value $x_0$, and number of iterations $N$. 
 \item For $t=0,\ldots,N-1$
\begin{enumerate}
 \item Generate $\e \sim g(\cdot)$ and $u \sim$ U$(0,1)$ independently
 \item If $0 < u < p$, set 
\[ x' = T(x_t,\e) \quad \textrm{ and } \quad \alpha(x_t,\e) 
= \min\left(1, \dfrac{1-p}{p} ~\dfrac{\pi(x')}{\pi(x_t)} ~J(x,\e) \right)\]
\item Else if $p < u < 1$ set
\[ x' = T^b(x_t,\e) \quad \textrm{ and } \quad \alpha(x_t,\e) 
= \min\left(1, \dfrac{p}{1-p} ~\dfrac{\pi(x')}{\pi(x_t)} ~\dfrac1{J(x,\e)} \right)\]
\item Set \[x_{t+1} = \left\{\begin{array}{ccc}
 x' & \textsf{ with probability } & \a(x_t,\e) \\
 x_{t}& \textsf{ with probability } & 1 - \a(x_t,\e)
\end{array}\right.\]
\end{enumerate}
\item End for
\end{itemize}
\botline \rmfamily
\end{algo}
Notably, the acceptance probability is independent of the distribution $g(\cdot)$, even if it is not symmetric.
The algorithm can be shown to be a special case of MH algorithm with the mixture proposal density:
\begin{equation}\label{eqn:mixturepro_equiv}
\begin{split}
 q(x \to z) & = p~g(T^{-1}(x,z))\left|\dfrac{\partial T^{-1}(x,z)}{\partial z}\right|\mathbb I(z \in T(x,\Y)) \\
& \qquad + (1-p)~g({T^b}^{-1}(x,z))\left|\dfrac{\partial {T^b}^{-1}(x,z)}{\partial z}\right|\mathbb I(z \in T^b(x,\Y))
\end{split}
\end{equation}
where the \emph{inverses} are defined by
\begin{enumerate}
\item $T(x,T^{-1}(x,z)) = z = T^{-1}(x,T(x,z)),~\forall~ z\in T(x,\Y)$
\item $T^b(x,{T^b}^{-1}(x,z)) = z = {T^b}^{-1}(x,T^b(x,z)),~\forall~ z\in T^b(x,\Y)$
 \end{enumerate}
In Section S-1 we show that detailed balance holds for the above algorithm. This ensures
that our TMCMC methodology has $\pi$ as the stationary distribution. 
Although in this univariate case TMCMC is an MH algorithm with the 
specialized mixture density (\ref{eqn:mixturepro_equiv}) as the proposal mechanism,
this proposal distribution becomes singular in general in higher dimensions.


We remark that TMCMC maybe particularly useful for improving
the mixing properties of the Markov chain. For instance, if there are distinct modes in 
several disjoint regions of state space, then standard MH algorithms tend to get trapped
in some modal regions, particularly if the proposal distribution has small variance.
Higher variance, on the other hand, may lead to poor acceptance rates in standard MH algorithms.
Gibbs sampling is perhaps more prone to mixing problems due to the lack of tuning facilities.
For multimodal target distributions, mixture proposal densities are often recommended. For
instance, \ctn{Guan07} theoretically prove that a mixture of two
proposal densities results in a ``rapidly mixing" Markov chain when the target distribution is multimodal. 
Our proposal, which we have shown to be a mixture density in the one-dimensional case,
seems to be appropriate from this perspective. 
Indeed, in keeping with this discussion, \ctn{Dutta10}, apart from
showing that the multiplicative transformation is geometrically ergodic even in situations
where the standard proposals fail to be so, demonstrated that it is very effective for bimodal
distributions. 
These arguments demonstrate that a real advantage of TMCMC (also of other transformation-based methods
as in \ctn{Liu01}) comes forth when the transformations 
associated with our method 
identify a subspace moving within which allows to explore regions that are 
otherwise separated by valleys in the probability
function. Efficient choice of transformations of course depends upon the target distribution.

In higher dimensions our proposal does not
admit a mixture form but since the principles are similar, it is not unreasonable to expect
good convergence properties of TMCMC in the cases of high-dimensional and/or multimodal target densities.
In the multidimensional case, which makes use of multivariate transformations (which we introduce next),
reasonable acceptance rates can also be ensured, in spite of
the high dimensionality. This  
we show in Section S-6, 
and illustrate with the Challenger data problem
and particularly with the geostatistical problem. 
Moreover, the multivariate transformation method brings out other significant 
advantages of our method, for instance, computational speed and the ability to overcome mixing 
problems caused by highly correlated variables.

\subsection{Multivariate case}\label{sec:gtm:multivar}
Suppose now that $\statesp$ is a $k$-dimensional space of the form $\statesp = \prod_{i=1}^k \statesp_i$ 
so that $T = (T_1,\ldots,T_k)$ where each $T_i : \statesp_i \times \mathcal D \to \statesp_i$, for some 
set $\mathcal D$, are transformations as in Section \ref{sec:gtm:univar}. 
Let $\bz=(z_1,\ldots,z_k)$ be a vector of indicator variables, where, 
for $i=1,\ldots,k$, 
$z_i=1$ and $z_i=-1$ indicate, respectively, application of forward transformation and  
backward transformation to $x_i$.
Given any such indicator vector $\bz$, let us define
$ T_{\bz} = (g_1,g_2,\ldots,g_k)$
where 
\[ g_i = \left\{ \begin{array}{ccc}
                  T_i^b & \textrm{ if } & z_i=-1 \\ 
		  T_i & \textrm{ if } & z_i=1. 
                 \end{array}
\right.\]
Corresponding to any given $\bz$, we also define the following `conjugate' vector
$\bz^c=(z^c_1,z^c_2,\ldots,z^c_k)$, where 
\[ z^c_i = \left\{ \begin{array}{ccc}
                  1 & \textrm{ if } & z_i=-1 \\ 
		  -1 & \textrm{ if } & z_i=1. 
                 \end{array}
\right.\]
With this definition of $\bz^c$, $T_{\bz^c}$ can be interpreted as the conjugate of $T_{\bz}$.

Since $2^k$ values of $\bz$ are possible, it is clear that $T$, via $\bz$,
induces $2^k$ many types of `moves' of the forms $\{T_{\bz_i};i=1,\ldots,2^k\}$ on the state--space. 
Suppose now that there is a subset $\Y$ of $\D$ such that the sets $T_{\bz_i}(\bm x,\Y)$ and 
$T_{\bz_j}(\bm x,\Y)$ are disjoint for every $\bz_i \ne \bz_j$.

\textbf{Examples: Transformations on higher dimensional state--space}
\begin{enumerate}


\item (Additive transformation) Suppose $\statesp = \D=\mathbb R^2$. With two
positive scale parameters $a_1$ and $a_2$, we can then consider the following
additive transformation:
$T_{(1,1)}(\bm x,\be) = (x_1 + a_1\e_1,x_2+a_2\e_2)$, 
$T_{(-1,1)}(\bm x,\be) = (x_1 - a_1\e_1,x_2 + a_2\e_2)$, 
$T_{(1,-1)}(\bm x,\be) = (x_1 + a_1\e_1,x_2 - a_2\e_2)$ and 
$T_{(-1,-1)}(\bm x,\be)= (x_1 - a_1\e_1,x_2 - a_2\e_2)$. 
We may choose $\Y = (0,\infty)\times (0,\infty)$. 

\item (Multiplicative transformation) Suppose $\statesp = \D = \mathbb R \times (0,\infty)$.
Then we may consider the following multiplicative transformation:
$T_{(1,1)}(\bm x,\pmb \e) = (x_1\e_1 , x_2\e_2)$,
$T_{(-1,1)}(\bm x,\pmb \e) = (x_1/\e_1,x_2\e_2)$, 
$T_{(1,-1)}(\bm x,\pmb \e) = (x_1\e_1,x_2/\e_2)$ and 
$T_{(-1,-1)}(\bm x,\pmb\e) = (x_1/\e_1,x_2/\e_2)$. 
We may let $\Y = \left\{(-1,1)-\{0\}\right\} \times (0,1)$.

\item (Additive-multiplicative transformation) Suppose $\statesp = \D = \mathbb R \times (0,\infty)$.
It is possible to combine additive and multiplicative transformations in the following manner:
$T_{(1,1)}(\bm x,\pmb \e) = (x_1 + \e_1 , x_2\e_2)$,
$T_{(-1,1)}(\bm x,\pmb \e) = (x_1 - \e_1,x_2\e_2)$, 
$T_{(1,-1)}(\bm x,\pmb \e) = (x_1 + \e_1,x_2/\e_2)$ and 
$T_{(-1,-1)}(\bm x,\pmb\e) = (x_1 - \e_1,x_2/\e_2)$. We may let $\Y = (0,\infty) \times (0,1)$.


\end{enumerate}

The above examples can of course be generalized to arbitrary dimensions.
Also, it is clear that it is possible to construct valid transformations in high-dimensional spaces
using combinations of valid transformations on one-dimensional spaces.

Now suppose that $g$ is a density on $\Y$, and, for $i=1,\ldots,2^k$, let $P_i=P(T_{\bz_i})$ 
be the probability of the move-type $T_{\bz_i}$. 
We assume that for each $i$, $P_i>0$
and $\sum_{i=1}^{2^k}P_i=1$. Note that this requires us to specify the $2^k$-dimensional probability
vector, which seems to be a daunting task for large $k$. However, in Section \ref{subsec:movetypeprob}
we show that this difficulty can be overcome by considering a product form of the move-type
probabilities induced by a mechanism of simulating $\bz$, which facilitates the choice of appropriate move-types
from the very large set of available move-types. This mechanism is also highly efficient computationally.

The MCMC algorithm based on transformations is given in 
Algorithm \ref{algo:GTM:multivar}.
\vspace{2cm}
\begin{algo}\label{algo:GTM:multivar} \topline MCMC algorithm based on transformation (multivariate case) \botline \normalfont \ttfamily
\begin{itemize}
 \item Input: Initial value $\supr{\bm x}{0}$, and number of iterations $N$. 
 \item For $t=0,\ldots,N-1$
\begin{enumerate}
 \item Generate $\pmb \e \sim g(\cdot)$ and an index $i \sim \mathcal M(1;P_1,\ldots,P_{2^k})$ independently.
 Actually, simulation from the multinomial distribution is not necessary; see Section \ref{subsec:movetypeprob}
 for an efficient and computationally inexpensive method of generating the index even when the number of move-types
 far exceeds $2^k$.
 \item \[ \bm x' = T_{\bz_i}(\supr{\bm x}{t},\pmb \e) \quad \textrm{ and } 
 \quad \alpha(\supr{\bm x}{t},\pmb \e) = \min\left(1, \dfrac{P(T_{\bz^c_i})}{P(T_{\bz_i})} ~\dfrac{\pi(\bm x')}{\pi(\supr{\bm x}{t})} 
 ~\left|\frac{\partial (T_{\bz_i}(\supr{\bm x}{t},\pmb \e),\pmb\e)}{\partial(\supr{\bm x}{t},\pmb \e)}\right| \right)\]
\item Set \[ \supr{\bm x}{t+1}= \left\{\begin{array}{ccc}
 \bm x' & \textsf{ with probability } & \a(\supr{\bm x}{t},\pmb\e) \\
 \supr{\bm x}{t}& \textsf{ with probability } & 1 - \a(\supr{\bm x}{t},\pmb\e)
\end{array}\right.\]
\end{enumerate}
\item End for
\end{itemize}
\botline \rmfamily
\end{algo}
In light of the above algorithm, it can be seen that for each of the transformations in the above examples, 
a mixture proposal of the form
(\ref{eqn:mixturepro_equiv}) is induced.
It will, however, be pointed out in Section \ref{sec:single_e} that a singleton $\e$ suffices
for updating multiple random variables simultaneously, which would imply singularity of the
underlying proposal distribution.
Notice that for arbitrary dimensions the additive transformation reduces to the RWMH.

Algorithm \ref{algo:GTM:multivar} indicates that updating highly correlated variables can be done 
naturally with TMCMC:
for instance, in Example 1 of this section one may select $T_{(1,1)}(\bm x,\be)$ and 
$T_{(-1,-1)}(\bm x,\be)$ with high probabilities if $x_1$ and $x_2$ are highly positively correlated
and $T_{(-1,1)}(\bm x,\be)$ and $T_{(1,-1)}(\bm x,\be)$ may be selected with high probabilities 
if $x_1$ and $x_2$ are highly 
negatively correlated.

\section{Validity and usefulness of singleton $\e$ in implementing TMCMC in high dimensions}
\label{sec:single_e}

Crucially, a singleton $\e$ suffices to ensure the validity of our algorithm, even though
many variables are to be updated. 
This indicates a very significant computational advantage over all
other MCMC-based methods: for instance, complicated simulation of hundreds of thousands of variables 
may be needed for any MCMC-based method, while, for the same problem, a single simulation
of our methodology will do. Indeed, 
in Section \ref{sec:appl:dtm} we update 160 variables using
a single $\e$ in the geostatistical problem of \ctn{Diggle98}. 
This singleton $\e$ also ensures that a mixture MH proposal density corresponding to our
TMCMC method does not exist. The last fact shows that TMCMC 
can not be a special case of the MH algorithm. On the other hand, assuming that
instead of singleton $\e$, there is an $\e_i$ associated with each of the variables $x_i$; $i=1,\ldots,k$,
then again TMCMC boils down to the MH algorithm, and, as in the univariate case, here also 
our transformations would induce a mixture 
proposal distribution for the algorithm, consisting of $2^k$ mixture components each corresponding to
a multivariate transformation.

Using singleton $\e$, for transformations other than the additive transformation, it is necessary to 
incorporate extra move types
having positive probability which change one variable using forward or backward transformation, keeping the other
variables fixed at their current values. Consider for instance,  Example 3 of Section \ref{sec:gtm:multivar}. The example
indicates that, with a singleton $\e$, it is only possible to move from $(x_1,x_2)$ to either of the following states: $(x_1+\e,x_2\e)$,
$(x_1-\e,x_2\e)$, $(x_1+\e,x_2/\e)$ and $(x_1-\e,x_2/\e)$ with positive probabilities. 
In addition, we could specify that the states $(x_1,x_2\e)$, $(x_1,x_2/\e)$,
$(x_1+\e,x_2)$ and $(x_1-\e,x_2)$ also have positive probabilities to be visited from $(x_1,x_2)$ in one step. 
We will need to specify the
visiting probabilities $P_i>0;i=1,\ldots,8$ such that $\sum_{i=1}^8P_i=1$. 
A general method of specifying the move-type probabilities, which also preserves computational efficiency,
is discussed in Section \ref{subsec:movetypeprob}.
Inclusion of the extra move types 
ensures irreducibility and aperiodicity (the definitions
are provided in Section S-3) of the Markov chain.
It is easy to see that even for higher dimensions irreducibility and aperiodicity can be enforced by bringing 
in move types of similar forms that updates one variable
keeping the remaining variables fixed. One only needs to bear in mind that the move types must be included in pairs, that is, a move type that updates only 
the $i$-th co-ordinate $x_i$ using forward transformation and the conjugate move type that updates only $x_i$ using the backward transformation
both must have positive probability of selection. 

With single $\e$ and the addition of the extra move types Algorithm \ref{algo:GTM:multivar} 
requires only slight modification. 
As in Section \ref{sec:gtm:multivar} let $\bz=(z_1,\ldots,z_k)$ be the vector of indicator variables, but
now, in addition to the values $1$ and $-1$ as before, $z_i$ can take the value $0$ as well, indicating
no change to $x_i$. The generalized definition of $z_i$ can be expressed as follows: 
\[ z_i = \left\{ \begin{array}{ccc}
                  1 & \textrm{indicates forward transformation to} & x_i \\ 
		  0   & \textrm{ indicates no change to} & x_i \\
		  -1 & \textrm{indicates negative transformation to} & x_i. 
                 \end{array}
\right.\]
Given any such indicator vector $\bz$, we define as before
$ T_{\bz} = (g_1,g_2,\ldots,g_k)$
where now we extend the definition of $g_i$ to the following: 
\[ g_i = \left\{ \begin{array}{ccc}
                  T_i^b & \textrm{ if } & z_i=-1 \\ 
		  x_i   & \textrm{ if } & z_i=0 \\
		  T_i & \textrm{ if } & z_i=1. 
                 \end{array}
\right.\]
We also need to extend the definition of the conjugate vector: given $\bz$, we 
define the conjugate vector
$\bz^c=(z^c_1,z^c_2,\ldots,z^c_k)$, where 
\[ z^c_i = \left\{ \begin{array}{ccc}
                  1 & \textrm{ if } & z_i=-1 \\ 
		  0   & \textrm{ if } & z_i=0 \\
		  -1 & \textrm{ if } & z_i=1. 
                 \end{array}
\right.\]
In this definition of $z_i$, $3^k$ values of $\bz$ are possible, so that
we now have $3^k$ possible move-types the forms $\{T_{\bz_i};i=1,\ldots,3^k\}$ on the state--space. 
Now note that the move type induced by $\bz=(0,0,\ldots,0)$ does not propose any change to the current
state $\bx$. Hence, we discard this move, and consider the remaining $3^k-1$ move-types for
our TMCMC methodology.
Now suppose that $g$ is a density on $\Y$, and, for $i=1,\ldots,3^k-1$, let $P_i=P(\bz_i)$ 
be the probability of the move-type $T_{\bz_i}$. We assume that for each $i$, $P_i>0$
and $\sum_{i=1}^{3^k-1}P_i=1$.

With these minor modifications Algorithm (\ref{algo:GTM:multivar}) goes through with $\pmb \e$
replaced by the singleton $\e$. For the sake of completeness, we present our general TMCMC
algorithm based on a single $\e$ in Section S-2 (Algorithm S-2.1).  

This strategy works for all transformations, including the examples in Section \ref{sec:gtm:multivar} where we now
assume equality of all the components of $\pmb\e$. 
Only additional move types are involved for transformations in general.  
However, we prove
in Section S-3 that the additive transformation
does not require the additional move types. Also taking account of the inherent simplicity of this transformation,
the additive transformation is our automatic choice for the applications reported in this paper. 

\subsection{Flexible and computationally efficient specification of the move-type probabilities}
\label{subsec:movetypeprob}

An apparent drawback of Algorithms \ref{algo:GTM:multivar} and S-2.1 
is the
difficulty of specifying the move-type probabilities $p(\bz)$ for all possible values of $\bz$.
For large dimension $k$, manual specification of such high-dimensional probability vector is clearly infeasible.
Moreover, step 1 of Algorithms \ref{algo:GTM:multivar} and S-2.1 
refers to
simulation from a multinomial distribution involving the very high-dimensional move-type probability vector. 
But simulation from such a high-dimensional multinomial distribution can be computationally burdensome
in the extreme if traditional methods of multinomial simulation are used, 
even if specification of the move-type probability vector is at all possible.
In this section we show how both these problems can be avoided. 
The key idea is to note that the move-type probabilities of $T_{\bz}$ can be induced by
assigning probabilities to all possible values of $\bz$; a simple, but useful way is to assign positive probabilities to 
$\{-1,0,1\}$, the possible values of each component $z_i$ of $z$. 
The latter induces a probability distribution on the set of available move-types $T_{\bz}$,
and hence on the high-dimensional multinomial distribution.
Simulation of $\bz$ by drawing $z_i$ independently for $i=1,\ldots,k$ yields the move-type
$T_{\bz}$, thus obviating the requirement of simulation from the high-dimensional multinomial
distribution using traditional methods. In this mechanism specification of only the probabilities  
$Pr(Z_i=1)$ and $Pr(Z_i=-1)$ for $i=1\ldots,k$, are required, which is manageable. 
Details follow.

Consider a $k$ $(\geq 1)$-dimensional target distribution, with associated random variables $\bx=(x_1,\ldots,x_k)$.
Then, 
we can implement the following simple rule. 
Given $\bx$, let the forward and the backward
transformations be applied to $x_i$ with probabilities $p_i$ and $q_i$, respectively. With probability $1-p_i-q_i$,
$x_i$ remains unchanged. We now define $\bz$ to be a random vector such that the random variable 
$z_i$ takes values
$-1,0,1$, with probabilities $q_i,1-p_i-q_i,p_i$, respectively. The values 
$-1,0,1$ correspond, as before, to backward transformation,
no change, and forward transformation, respectively.

This rule, which is to be applied to each of $i=1,\ldots,k$ coordinates, includes
all possible move types, including the one where none of the $x_i$ is updated, that is, $\bx$ is taken to $\bx$. 
Since
the move-type $\bx\mapsto\bx$ is redundant, this is to be rejected whenever it appears. In other words, 
we would keep simulating the discrete random vector $\bz=(z_1,\ldots,z_k)$ until at least one $z_i\neq 0$,
and would then select the corresponding move type. For any dimension, this is a particularly simple and computationally
efficient exercise, since the rejection region is a singleton, and has very small probability (particularly
in high dimensions) if either of $p_i$ and $q_i$ is high for at least one $i$.

Since now we induce the probability distribution of $T_{\bz}$ through $\bz$, we denote $P(T_{\bz})$
by $P(\bz)$.
The above method implies that the probability of a move-type, given $\bz$, is of the form 
$$P(\bz)=C\underset{\{i_1:z_{i_1}=1\}}\prod p_{i_1}\underset{\{i_2:z_{i_2}=-1\}}\prod q_{i_2}
\underset{\{i_3:z_{i_3}=0\}}\prod(1-p_{i_3}-q_{i_3}),$$
and $C$ is the normalizing constant, which arose due to
rejection of the move type $\bx\mapsto\bx$. This normalizing constant cancels in the acceptance
ratio, and so it is not required to calculate it explicitly, another instance of preservation
of computational efficiency. 
Note that the probability of the conjugate move-type is
\begin{eqnarray*}
P(\bz^c)&=&C\underset{\{i_1:z^c_{i_1}=1\}}\prod p_{i_1}\underset{\{i_2:z^c_{i_2}=-1\}}\prod q_{i_2}
\underset{\{i_3:z^c_{i_3}=0\}}\prod (1-p_{i_3}-q_{i_3})\\
&=&C\underset{\{i_1:z_{i_1}=-1\}}\prod p_{i_1}\underset{\{i_2:z_{i_2}=1\}}\prod q_{i_2}
\underset{\{i_3:z_{i_3}=0\}}\prod (1-p_{i_3}-q_{i_3}),
\end{eqnarray*}
so that the factor $\underset{\{i_3:z_{i_3}=0\}}\prod (1-p_{i_3}-q_{i_3})$ cancels in the acceptance ratio, further
simplifying computation.

Algorithm \ref{algo:GTM:multivar3} gives the simplified TMCMC algorithm based on a singleton $\e$.
\begin{algo}\label{algo:GTM:multivar3} \topline Simplified TMCMC algorithm based on a single $\e$.
\botline \normalfont \ttfamily
\begin{itemize}
 \item Input: Initial value $\supr{\bm x}{0}$, and number of iterations $N$. 
 \item For $t=0,\ldots,N-1$
\begin{enumerate}
 \item Generate $\e \sim g(\cdot)$ and simulate $\bz$ by generating 
 $z_i \sim \mathcal M(1;p_i,q_i,1-p_i-q_i)$ independently
 for $i=1,\ldots,k$.  
 \item \[ \bm x' = T_{\bz}(\supr{\bm x}{t}, \e) \quad \textrm{ and } 
 \quad \alpha(\supr{\bm x}{t}, \e) = \min\left(1, \dfrac{P(\bz^c)}{P(\bz)} ~\dfrac{\pi(\bm x')}{\pi(\supr{\bm x}{t})} 
 ~\left|\frac{\partial (T_{\bz}(\supr{\bm x}{t}, \e),\e)}{\partial(\supr{\bm x}{t}, \e)}\right| \right),\]
 where 
 \[\dfrac{P(\bz^c)}{P(\bz)}=\underset{\{i_1:z_{i_1}=-1\}}\prod \frac{p_{i_1}}{q_{i_1}}
 \underset{\{i_2:z_{i_2}=1\}}\prod \frac{q_{i_2}}{p_{i_2}}.\]

\item Set \[ \supr{\bm x}{t+1}= \left\{\begin{array}{ccc}
 \bm x' & \textsf{ with probability } & \a(\supr{\bm x}{t},\e) \\
 \supr{\bm x}{t}& \textsf{ with probability } & 1 - \a(\supr{\bm x}{t},\e)
\end{array}\right.\]
\end{enumerate}
\item End for
\end{itemize}
\botline \rmfamily
\end{algo}

For the additive transformation, the issues are further simplified. The random variable $z_i$ here takes
the value $-1$ and 1 with probabilities $p_i$ and $q_i=1-p_i$, respectively. So, only $p_i$ needs
to be specified. Since $z_i=0$ has probability zero in this setup, there is no need to perform
rejection sampling to reject any move-type. 

\subsection{Discussion on choices of $p_i$ and $q_i$}
\label{subsec:probchoice}

Interestingly, the ideas developed in Section \ref{subsec:movetypeprob} provide us with a handle to control the move-type probabilities, by simply
controlling $p_i$ and $q_i$ for each $i$. For instance, if some pilot MCMC analysis tells us that
$x_i$ and $x_j$ are highly positively correlated, then we could set $p_i$ and $p_j$ (or $q_i$ and $q_j$)
to be high provided the forward transformation on both $x_i$ and $x_j$ are increasing. 
On the other hand, if $x_i$ and $x_j$ are highly negatively correlated, then we can set $p_i$ to be 
high (low) and $q_j$ to be low (high) and so on. Apart from these choices, there are theoretically motivated
choices of $p_i$ and $q_i$ as well. Indeed, \ctn{Dey13a} prove, under
suitable regularity conditions, that additive TMCMC is geometrically ergodic when
$p_i=q_i=1/2$. Thus, at least for additive transformations, the choice $p_i=q_i=1/2$ for $i=1\ldots,k$,
seems to be reasonable from a theoretical perspective. In our TMCMC illustration of the Challenger data 
presented in Section \ref{sec:appl:logistic}
we choose $p_i,q_i$ based on the posterior correlations obtained from a pilot MCMC analysis, whereas in
the case of Rongelap data we set $p_i=q_i=1/2$.

\subsubsection{Dependence structure on $\bz$}
\label{subsubsec:dependent_z}

The procedure outlined above simulates each co-ordinate $z_i$ independently, for $i=1,\ldots,k$. 
But because the same $\e$ is used for the transformation of each co-ordinate $x_i$ of $x$, the co-ordinate
moves are dependent. However in addition, it is also possible to consider dependence between the components of $\bz$
using a hierarchical structure.
For example, for $i=1,\ldots,k$, let $\bw_i=(w_{i1},\ldots,w_{ik})\sim N(\bmu_i,\bSigma_i); i=1,2,3$,
where the parameters $\left(\bmu_i,\bSigma_i\right);i=1,2,3$ are assumed to be known.
We then set $p_i=\exp\left(w_{1i}\right)/\sum_{j=1}^3\exp\left(w_{ji}\right)$, 
$q_i=\exp\left(w_{2i}\right)/\sum_{j=1}^3\exp\left(w_{ji}\right)$, so that
$1-p_i-q_i = \exp\left(w_{3i}\right)/\sum_{j=1}^3\exp\left(w_{ji}\right)$.
These $k$-variate normal distributions induce dependence between
$\bp=(p_1,\ldots,p_k)$ and $\bp=(q_1,\ldots,q_k)$.
Thus, even though conditionally on $\{(p_i,q_i);i=1,\ldots,k\}$ $z_i$ are independent, marginalized
over $\bp$ and $\bq$, the components of $\bz$ are dependent. 
To achieve the effect
of this dependent structure in TMCMC in a theoretically valid manner, 
at each iteration of the TMCMC algorithm we can simulate 
$\bw_1,\bw_2,\bw_3$ from their respective $k$-variate normal distributions, and
from the simulated values obtain, for $i=1,\ldots,k$, 
$p_i=\exp\left(w_{1i}\right)/\sum_{j=1}^3\exp\left(w_{ji}\right)$, and 
$q_i=\exp\left(w_{2i}\right)/\sum_{j=1}^3\exp\left(w_{ji}\right)$. 
To avoid getting $p_i$ and $q_i$ too close to zero in some simulations, we can appropriately truncate the
$k$-variate normal distributions.
Once $\{(p_i,q_i);i=1,\ldots,k\}$
are obtained, conditionally on these probabilities, $z_i;i=1,\ldots,k$ will be simulated independently.
Thus, the algorithm for dependent $\bz$ admits the same form as Algorithm \ref{algo:GTM:multivar3}; only
in the first step, simulation of $\bp$ and $\bq$ from their respective
dependent distributions must precede independent simulation of $z_i;i=1,\ldots,k$. The algorithm 
(Algorithm S-4.1) and
proof of detailed balance are provided in Sections S-4 and S-5, respectively.

Note that for the additive transformation, since $q_i=1-p_i$, only the joint distribution of $\bp$
needs to be considered, with $p_i=\exp\left(w_{1i}\right)/\sum_{j=1}^2\exp\left(w_{ji}\right)$,
and it is not necessary to simulate $\bw_3$ at all.
 
Theoretically appropriate choices of $\left(\bmu_i,\bSigma_i\right);i=1,2,3$ will be our topic of future research
but from a practical point of view, one can tune these parameters 
to achieve good mixing properties and acceptance rates. 

\subsection{Advantages of TMCMC updating with single $\e$}
\label{subsec:advantages_single_e}

Standard methods like sequential RWMH may tend to be computationally
infeasible in high dimensions while inducing mixing problems due to
posterior dependence between the parameters, whereas TMCMC remains 
free from the aforementioned problems thanks to singleton $\e$ and joint updating
of all the parameters.
Specialized proposals for joint updating may be constructed for specific problems only;
for instance, block updating proposals for Gaussian Markov random fields are available
(\ctn{Rue01}). But generally, efficient block updating proposals are not available. Moreover, even in the specific
problems, simulation from the specialized block proposals and calculating
the resulting acceptance ratio are generally computationally very expensive.
In contrast, TMCMC with singleton $\e$ seems to be much more general and efficient.
Moreover, we demonstrate in Section \ref{sec:appl:logistic} in connection with the Challenger
data problem that TMCMC can outperform
well-established block proposal mechanisms, usually based on the asymptotic covariance matrix of the
maximum likelihood estimator (MLE), in terms of acceptance rate.


\section{Application of TMCMC to the Challenger dataset}\label{sec:appl:logistic}
In 1986, the space shuttle Challenger exploded during take off, killing the seven astronauts aboard. 
The explosion was the result of an O-ring failure, a splitting of a ring of rubber that seals the parts of 
the ship together. The accident was believed to be caused by the unusually cold weather ($31^0$F or $0^0$C) 
at the time of launch, as there is reason to believe that the O-ring failure probabilities increase as 
temperature decreases. The data are provided in Table S-1 
for ready reference. 
We shall analyze the data with the help of well--known logit model. Our main aim is not analyzing and 
drawing inference since it is done already in \ctn{Dalal89}, \ctn{Martz92} and \ctn{Robert04} . 
We shall rather compare the different MCMC methodologies used in Bayesian inference for logit--model. 
Let \[ \eta_i = \b_1 + \b_2x_i \] where $x_i = t_i/\max ~t_i$, $t_i$'s being the temperature at flight 
time (degrees F), $i=1,\ldots,n$. and $n=23$. Also suppose $y_i$ is the indicator variable denoting failure 
of 0-ring. We suppose $y_i$'s independently follow Bernoulli($\pi(x_i)$).

In the logit model we suppose that the log-odd ratio is a linear function of temperature at flight time, i.e.,
\[ \log\frac{\pi}{1-\pi} = \eta = \beta_1 + \beta_2x \]
which gives \[ \pi_i = \exp(\eta_i)\big/(1+\exp(\eta_i)).\]
In the absence of information regarding $(\beta_1,\beta_2)$, we specify a uniform (improper) 
prior for $(\beta_1,\beta_2)$.

We construct an appropriate additive transformation $T : \mathbb R^2 \times \mathbb R\to \mathbb R^2$ as follows.
First, we consider the form $T_{(1,1)}((\beta_1,\beta_2),\e)=(\beta_1,\beta_2)'+ (s_1\e,s_2\e)'$, where $s_1$ and $s_2$ are the standard errors of the 
maximum likelihood estimator of $(\beta_1,\beta_2)'$.
Thus, we finally obtain the transformation
\[ T_{(1,1)}((\b_1,\b_2),\e)  =  (\b_1 + 7.3773\e,\b_2 + 4.3227\e)\]
and use Algorithm \ref{algo:GTM:multivar3} with $\Y = (0,\infty)$ and 
\[ g(\e) \varpropto \exp(-\e^2/2), ~\e > 0\]
that is, the $N(0,1)$ distribution truncated to the left at zero. From the covariance matrix $\bC$ we observe that 
the correlation of $\hat \b_1$ and $\hat\b_2$ is approximately $-0.99$ and hence from our discussion in Section \ref{subsec:probchoice},
setting high probabilities to the moves $T_{(1,-1)}(\bx,\e)$ and 
$T_{(-1,1)}(\bx,\e)$ should facilitate good mixing.  Following the discussion in Section \ref{subsec:probchoice} 
we set 
$P((1,1)) = P((-1,-1)) = 0.01$ and $P((1,-1)) = P((-1,1)) = 0.49$. 

Also for comparison we use the RWMH algorithm (both joint and sequential updation) and also the MH algorithm 
with proposal
$q(\pmb\b'|\pmb\b) = N(\pmb\beta,\bSigma)$ where $\bSigma = h^2\bC$ (we take $h=1$ for our purpose) 
with $\bC$ being the large sample covariance 
matrix of the MLE $\widehat{\pmb\b}$ of $\pmb\b$.
Table \ref{tab:challenger:analysis} gives the posterior summaries and Figure \ref{fig:challenger:trace} 
gives the trace plots of $\b_1$ and $\b_2$ for TMCMC sampler and the MH sampler. It is seen that the mixing 
is excellent even though a single $\e$ has been used.

\begin{table}
\begin{tiny}
\begin{tabularx}{\textwidth}{Z|Z|Z|Z|Z|Z|Z|Z|Z|Z} \hline \hline
{variable} & method & {acceptance rate (\%)} & mean   &  std & 2.5\%*  & 25\%* &  50\%*  &  75\%*  & 97.5\%* \\ \hline \hline
	& RWMH  &    42.17         & 19.119  & 8.078 & 4.909 & 13.481 & 18.475 & 24.227 & 38.176  \\ \cline{2-10}
$\b_1$	& MH &       42.60         & 18.930  & 8.513 & 5.011 & 12.823 & 17.981 & 23.957 & 38.206 \\ \cline{2-10}
	& TMCMC &      73.23         & 18.973  & 7.944 & 4.970 & 12.881 & 16.210 & 21.685 & 37.877 \\ \hline \hline
	& RWMH &     48.14         & -23.724 & 9.613 & -46.272  & -29.786 & -22.984 & -17.019 & -6.7792\\ \cline{2-10}
$\b_2$	& MH &       42.60**        & -23.491 & 10.128 & -46.461 & -29.464 & -22.353 & -16.261 & -6.956 \\ \cline{2-10}
	& TMCMC &      73.23**        & -23.165 &  9.762 & -46.404 & -28.891 & -22.282 & -16.446 & -7.026 \\ \hline \hline
\end{tabularx}
\end{tiny}
\caption{\small Summary of the posterior samples based on MCMC runs of length 100,000 out of which first 20,000 samples 
are discarded as burn-in.} {\small RWMH = Random walk Metropolis-Hastings, MH = Metropolis-Hastings with bivariate normal 
proposal, TMCMC = MCMC based on transformation \\ * : posterior sample quantiles. \\ **: same as acceptance ratio 
for  $\beta_1$ since updated jointly. }
\label{tab:challenger:analysis}
\end{table}
\begin{figure}
 \centering\includegraphics[width=\textwidth]{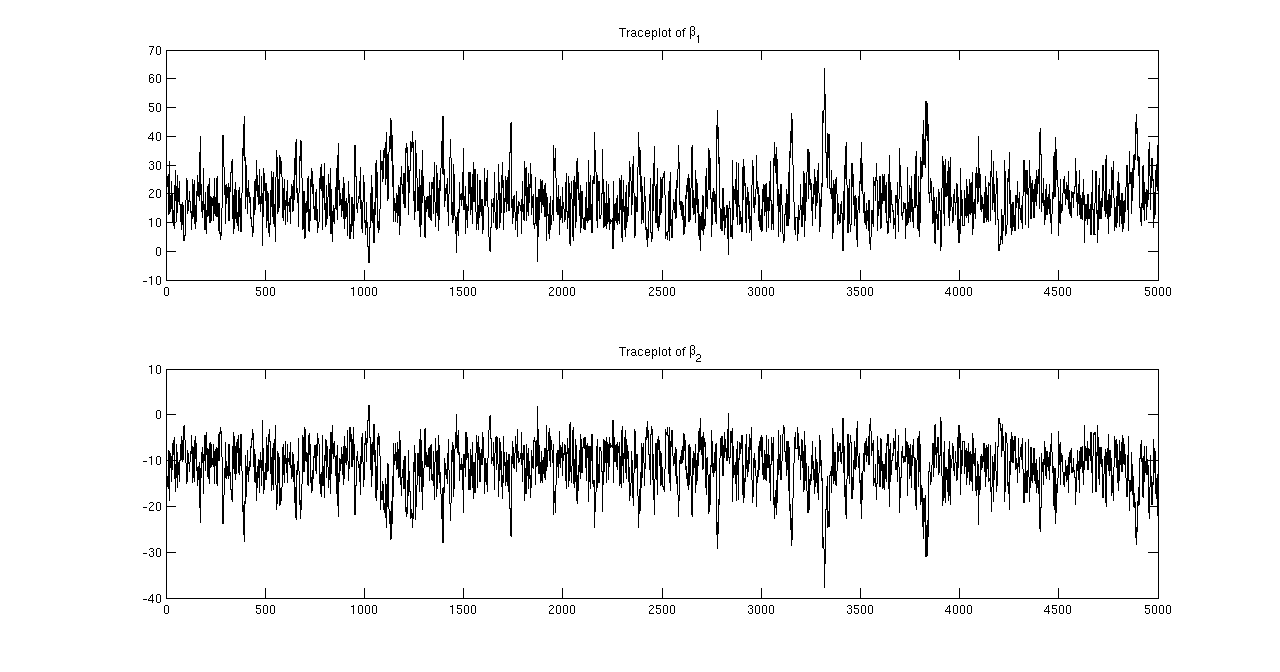}
(a)
\centering\includegraphics[width=\textwidth]{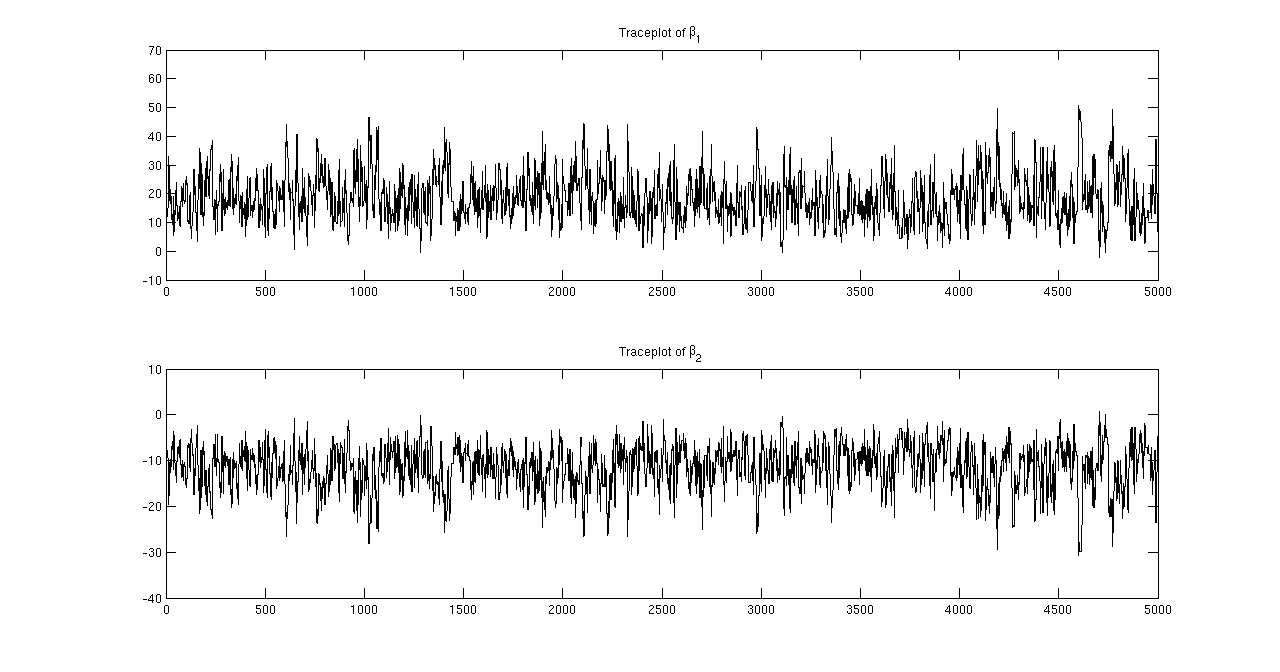}
(b)
\caption{Trace plots of $\b_1$ and $\b_2$ (a) TMCMC (b) MH}
\label{fig:challenger:trace}
\end{figure}

Notice the excellent result of the MCMC based on transformations. The acceptance ratio is almost twice as 
large as those for other two MH algorithms. 
As remarked in Section \ref{subsec:advantages_single_e}, 
indeed TMCMC outperformed the
MH block proposal based on the large sample covariance matrix of the MLE of $\pmb\b$ in terms
of acceptance rate.
Also for implementing TMCMC 
we need to simulate only one $\e$ in each step. In the RWMH with sequential updating and in MH based on 
bivariate normal proposal we need two such $\epsilon$'s. In the RWMH we need to calculate the likelihood 
twice in each iteration. So, TMCMC dominates the other two in this respect.
It can be easily anticipated, in light of the theoretical arguments regarding acceptance rate presented in Section S-6, 
that for joint RWMH the acceptance rate would be even lower.
In Section \ref{sec:appl:dtm}, where we consider a 160-dimensional problem, we show that,
that TMCMC outperforms joint RWMH 
by a substantially large margin in terms
of acceptance rate.

\section{Application of TMCMC to the geostatistical problem of radionuclide concentrations on Rongelap Island}
\label{sec:appl:dtm}

\subsection{Model and prior description}
\label{subsec:model_prior}
We now consider the much analyzed radionuclide count data on Rongelap Island
(see, for example, \ctn{Diggle97}, \ctn{Diggle98}, \ctn{Chris04}, \ctn{Chris06}), and illustrate
the performance of TMCMC with a singleton $\e$. For $i=1,\ldots,157$, 
\ctn{Diggle98} model the count data as
$$Y_i\sim Poisson(M_i),$$ where $$M_i=t_i\exp\{\beta+S(\bx_i)\};$$ $t_i$ is the duration of observation
at location $\bx_i$, $\beta$ is an unknown parameter and $S(\cdot)$ is a zero-mean Gaussian process
with isotropic covariance function of the form 
$$Cov\left(S(\bx^\star_1),S(\bx^\star_2)\right)
=\sigma^2\exp\{-\left(\alpha\parallel\bx^\star_1-\bx^\star_2\parallel\right)^{\delta}\}$$
for any two locations $\bx^\star_1,\bx^\star_2$. In the above, $\parallel\cdot\parallel$ denotes
the Euclidean distance between two locations, and $(\sigma^2, \alpha, \delta)$ are unknown parameters.
Typically in the literature $\delta$ is set equal to 1 (see, {\em e. g.} \ctn{Chris06}), which we
adopt. We assume uniform priors on the entire parameter space corresponding to
$(\beta, \log(\sigma^2), \log(\alpha))$. 

We remark that since the Gaussian process $S(\cdot)$ does not define a Markov random field,
the block updating proposal developed by \ctn{Rue01} is not directly applicable here.
\ctn{Rue09} attempt to develop deterministic approximations to latent Gaussian models,
but the scope of such approximations is considerably restricted by the
conditional independence (Gaussian Markov random field) assumption (\ctn{Banerjee09}). 
Thanks to the generality and efficiency of our proposed methodology, it seems most appropriate 
to fit the Rongelap island 
model using TMCMC with singleton $\e$.

\subsection{Results of additive TMCMC with singleton $\e$}
\label{subsec:tmcmc_results}

Drawing $\e\sim N(0,1)\mathbb I(\e>0)$, we considered the following additive transformation 
\begin{align}
T(\beta,\e)&=\beta\pm 2\e,\notag\\
T(\log(\sigma^2),\e)&=\log(\sigma^2)\pm 5\e,\notag\\
T(\log(\alpha),\e)&=\log(\alpha)\pm 5\e,\notag\\
T(S(\bx_i),\e)&=S(\bx_i)\pm 2\e; \ \ \mbox{for} \ \ i=1,\ldots,157\notag
\end{align}
The scaling factors associated with $\e$ in each of the transformations are chosen
on a trial-and-error basis after experimenting with several initial (pilot) runs of TMCMC.
We assigned equal probabilities to all the $2^{160}$ move types. 
Move types are selected by independently generating $z_i$ taking the values $+1$ and $-1$ with equal probabilities,
that is, we set $p_i=q_i=1/2$ for $i=1,\ldots,160$. As mentioned in Section \ref{subsec:probchoice} this choice
of equal probabilities of forward and backward transformation is motivated by our result on
geometric ergodicity.

After discarding the first $2\times 10^7$ iterations as burn-in, we stored 
1 in every 100 iterations in the next $3.5\times 10^7$ iterations. 
This entire simulation took about a week to run on an ordinary laptop machine
and about 3 days on a workstation.
The autocorrelation functions of the variables (after further thinning by 10) of our TMCMC run, displayed in   
Figure \ref{fig:tmcmc_dtm}, indicates reasonable mixing properties. 
 The acceptance rate, after discarding the burn-in period, is 
0.43\% (considering the complete run of TMCMC after burn-in, that is, including thinning as well).
\begin{figure}
 \centering
\includegraphics[width=0.4\textwidth]{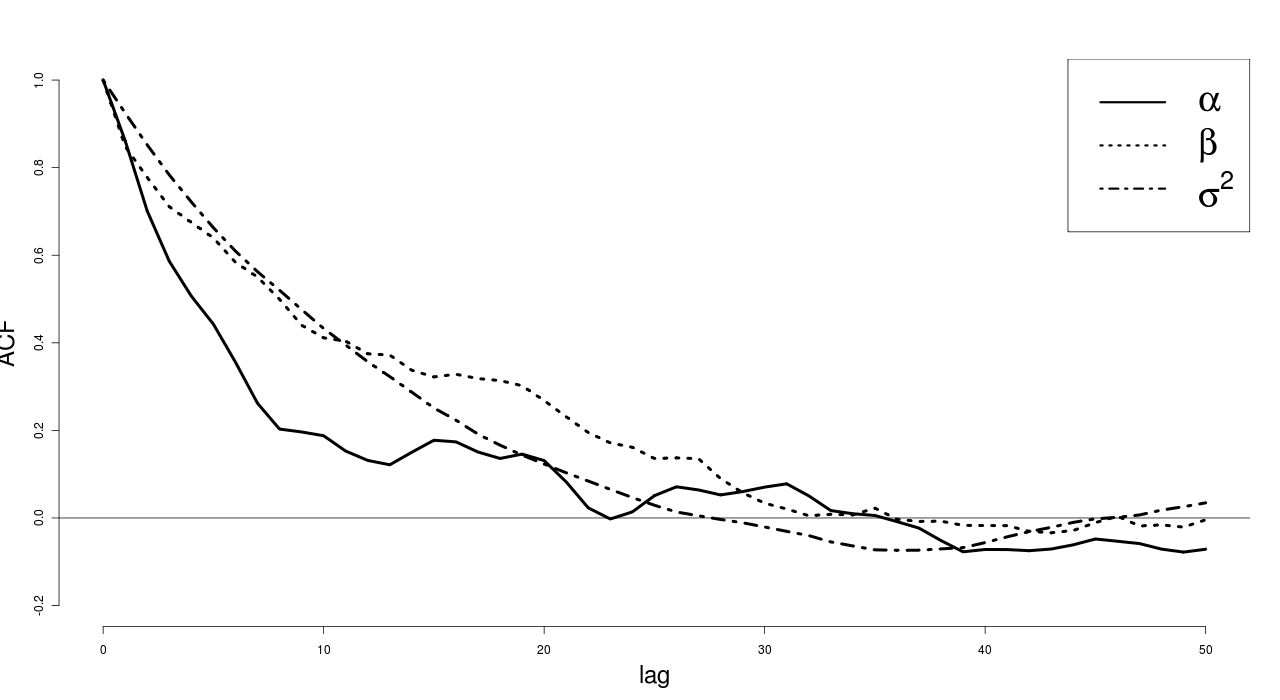}
\includegraphics[width=0.4\textwidth]{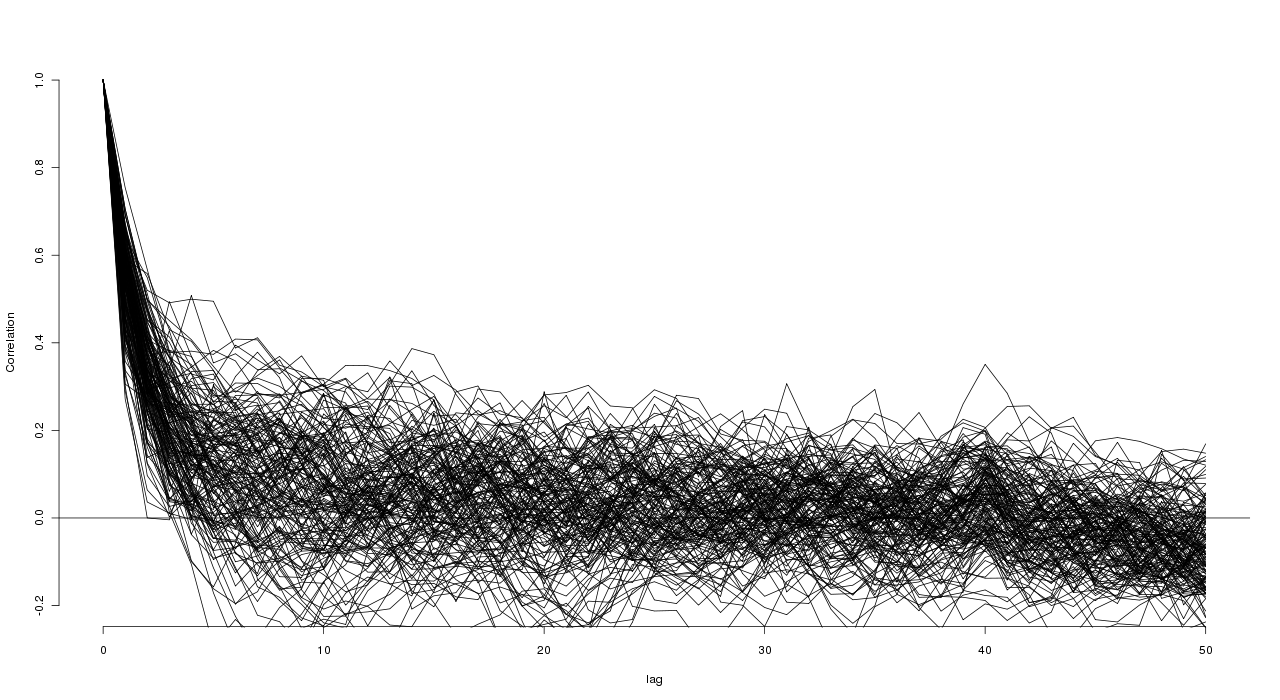}
\caption{Autocorrelation plots of the variables, $\alpha$,$\beta$, $\log\sigma^2$ (left-panel) and $s_{1},\ldots,s_{157}$ (right-panel) in the TMCMC run.}
\label{fig:tmcmc_dtm}
\end{figure}

\subsection{Comparison with joint RWMH}
\label{subsec:comparison_RWMH}

We also implemented a joint RWMH using the same additive transformation
as in Section \ref{subsec:tmcmc_results} but with different $\e$'s for each unknown.
Now the acceptance rate reduced to 0.0005\%. 
These observations are broadly in keeping with the theoretical discussions 
presented in Section S-6. 

\section{Application of TMCMC to doubly--intractable problem}\label{sec:doublyintract}
Doubly--intractable distributions arise quite frequently in fields like circular statistics, directed graphical models, Markov point processes etc. Even some standard distributions like gamma and beta involve intractable normalizing constants. 
Formally, a density $h(\by|\t)$ of the data set $\by=(y_1,\ldots,y_n)'$ is said to be doubly--intractable if it is of the form
\[ h(\by|\t) = f(\by|\t)/Z(\t)\]
where $Z(\t)$ is a function that is not available in closed form. So if we put a prior $\pi(\t)$ 
on $\t$, then the posterior is given by
\[ \pi(\t|\by)  = \dfrac{1}{c(\by)}~\dfrac{f(\by|\t)}{Z(\t)}\pi(\t) \quad \textrm{ where } \quad c(\by) = \dint_\Theta \dfrac{f(\by|\t)}{Z(\t)}\pi(\t)~d\t \]
Thus, if we try to apply MH like algorithms then the acceptance ratio will involve ratio of the 
function $Z(\cdot)$ at two parameter points $\t$ and $\t'$. Hence directly applying MH may not be feasible. 
Works by \cite{Moller04} and \cite{Murray06} are significant in this field. A double MH sampler approach 
is taken in \cite{Liang10}. In this section we briefly discuss the bridge--exchange algorithm by 
\cite{Murray06} and show how our application of TMCMC in the bridge--exchange algorithm may facilitate 
fast computation.

Suppose $M \in \mathbb N$ is the \emph{bridge size}, $\b_m = m/(M+1),~m=0,\ldots,M$. Define the density
\[ p_m(\bx|\t,\t') \propto f(\bx | \t)^{\b_m}f(\bx|\t')^{1-\b_m} \equiv f_m(\bx|\t,\t'), \quad m = 0,\ldots,M.\]
Obviously, $\bx$ is of the same dimensionality as $\by$; that is, $\bx=(x_1,\ldots,x_n)'$.
Further suppose that for each $m$, $~T_m(\bx\to \bx'|\t,\t')$ is a kernel satisfying the detailed balance condition
\[ T_m(\bx\to \bx'|\t,\t')p_m(\bx|\t,\t') = T_m(\bx'\to \bx|\t,\t')p_m(\bx'|\t,\t').\]
Now with a proposal density $q(\t\to\t'|\by)$ for the parameter, the bridge--exchange algorithm is given 
below.
\begin{algo}\label{algo:bridgexchange}\topline
The bridge--exchange algorithm \botline \normalfont \ttfamily 
\begin{itemize}
\item Input: initial state $\t_0$, length of the chain $N$, \#bridge \\ levels $M$.
\item \textbf{For} $t = 0,\ldots,N-1$
\begin{itemize}
 \item[1.] Propose $\t' \sim  q(\t' \from \t_t|\by)$
 \item[2.] Generate an auxiliary variable with exact sampling: \[ \bx_0 \sim p_0(\bx_0|\t,\t') \equiv f(\bx_0|\t')/Z(\t)\]
 \item[3.] Generate M further auxiliary variables with transition \\ operators:
\begin{eqnarray*}
 \bx_1 & \sim & T_1(\bx_0\to \bx_1 |\t,\t') \\
\bx_2 & \sim & T_2(\bx_1\to \bx_2 |\t,\t') \\
& \vdots & \\
\bx_M & \sim & T_M(\bx_{M-1}\to \bx_M |\t,\t') \\
\end{eqnarray*}
\item[4.] Compute $$\a(\t' \from \t_t) = \dfrac{q(\t'\to \t|\by) \pi(\t') f(\by|\t')}{q(\t\to \t'|\by) \pi(\t) f(\by|\t)} \prod_{m=0}^M\dfrac{f_{m+1}(\bx_m|\t,\t')}{f_{m}(\bx_m|\t,\t')}$$
 \item[5.] Set \[\t_{t+1} = \left\{\begin{array}{ccc}
 \t' & \textsf{ with probability } & \a(\t' \from \t_{t}) \\
 \t_{t}& \textsf{ with probability } & 1 - \a(\t' \from t_{t})
\end{array}\right.\]
\end{itemize}
\item \textbf{end for}
\end{itemize}
\rmfamily
\botline
\end{algo}

Now we see that, since each of the auxiliary variables $\bx_m,~m=1,\ldots,M$, is $n$-dimensional, generation of these
auxiliary variables may be computationally 
demanding if the sample size $n$ is moderate or large especially when one has to simulate from the 
sample space using accept-reject algorithms as in the case of circular variables. 
For any kernel $T_m$ which is not based on TMCMC, $O(nM)$ variables are required to be
generated from the state--space per iteration. Appealing to TMCMC, recall that with the 
additive transformation with a single $\e$, the kernel still satisfies the detailed balance condition. 

We assume that $\statesp$ is a group under some binary operation and that 
there is a homomorphism from $(\R^p,+)$ to $\statesp$ for some $p \in \mathbb N$. So we denote the 
binary operation on $\statesp$ by `$+$' itself. Let $g$ be a density on $\statesp$. 
We construct the kernels $T_m$ as follows:

\begin{algo}\label{algo:T_k}\topline
Construction of $T_m$ \botline \normalfont \ttfamily 
\begin{itemize}
 \item[1.] Generate $\e \sim g(\e)$ and $\bz\sim P(\bz)$, where $P(\bz)$ is some suitable distribution
 of $\bz$; $P(\bz)$ can be chosen as in the different versions of our TMCMC algorithm.
 \item[2.] Define the vector $\bx'$ by $$x_i' = \begin{cases}
                    x_{m-1,i} + a_i\e & \textrm{ if } ~z_i=1 \\ x_{m-1,i} - a_i\e & \textrm{ if } ~z_i=-1
                   \end{cases}$$
 \item[3.] Set $\a(\bx_{m-1} \to \bx') = \min\left( \dfrac{P(\bz^c)~f_m(\bx'|\t,\t')}{P(\bz)~f_m(\bx|\t,\t')},1\right)$
 \item[4.] Set $$\bx_{m} = \begin{cases}
\bx' & \textrm{ with probability } ~\a(\bx_{m-1} \to \bx') \\
\bx_{m-1} & \textrm{ with probability } ~1 - \a(\bx_{m-1} \to \bx')   \end{cases}$$ 
\end{itemize}

\rmfamily
\botline
\end{algo}

In this way we need only $O(M)$ simulations per iteration. Homomorphism from $(\R^p,+)$ to $\statesp$ holds in many cases, for example, in circular models where the state--space is $(-\pi,\pi]$ is a group with respect to addition modulo $\pi$.

\begin{figure}
 \centering\includegraphics[width=0.9\textwidth]{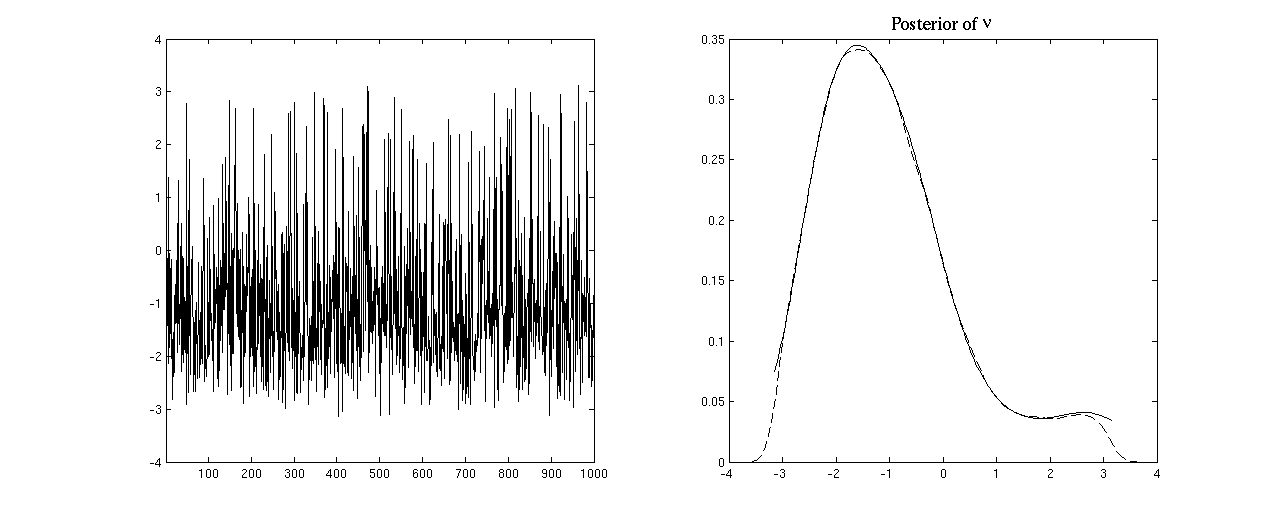}
\caption{Left panel: Trace plot of last 1,000 samples. Right panel: exact posterior density  of $\nu$ (solid line) and it's estimate (dash-dotted line).}
\label{fig:bridgetmcmc_nu}
\end{figure}

\subsection{Simulation study to illustrate TMCMC in bridge-exchange algorithm} 
Here we illustrate our method for a circular model of the form
\[ h(y|\nu) = \dfrac{1}{Z(\nu)}\exp(\cos(y + \nu\sin(y))),~ -\pi < y,\nu \leq \pi,\]
We generate a sample of size 20 from $h(y|\nu=0)$ and estimate the parameter 
$\nu$ based on this sample. The prior chosen on $\nu$ is the uniform distribution on $(-\pi,\pi]$
and $g(\cdot)$ is chosen to be the normal distribution with mean 0 and variance 1 restricted on the set 
$(0,\pi]$. Since the components of $\bx_0$ are $iid$, we used 
$Pr(Z_i=1) = Pr(Z_i=-1)=1/2$ and $a_i = 1$ for each $i$. 
We set $M = 100$ and chose $q(\nu'|\nu)$ to be the Von-mises distribution with mean $\nu$ and 
concentration 0.5 to keep the acceptance level around 63\%.

 The right panel of Figure \ref{fig:bridgetmcmc_nu} shows that the estimated posterior density  of 
 $\nu$ is very close to the exact posterior density. The little discrepancy at the tails are due to 
 the fact that $\nu$ is a circular variable and hence its support is $(-\pi,\pi]$ and the density 
 is \emph{not} zero at the end points -- a fact that is not incorporated in the kernel density estimator. 
 The left panel of the same figure shows that the mixing is excellent. Notice that here we have saved 
 $100(nM - M)/nM$ = 95\% simulations.

\section{Summary, conclusions and future work}
\label{sec:conclusions}

In this paper we have proposed a novel MCMC method that uses deterministic transformations
and move types to update the Markov chain. 
We have shown that our algorithm TMCMC generalizes the MH algorithm boiling down to MH
with a specialized proposal density in one-dimensional cases.
For higher dimensions if each component $x_i$ of the random vector to be updated
is associated with a distinct $\e_i$, then TMCMC again boils down to the MH algorithm 
with a specialized proposal density. But in dimensions greater than one, with less number of distinct $\e_i$
than the size of the random vector to be updated, TMCMC does not admit any MH representation.
Several versions of TMCMC have been detailed in this paper and in the supplement. In Section S-6
of the supplement, under reasonable regularity conditions we have provided and compared the 
asymptotic forms of the acceptance rates of RWMH and additive TMCMC when the dimensionality increases
to infinity
and have shown that the latter converges to zero at a much slower rate. 
That HMC is also a special case of TMCMC, is also explained in Section S-7; in addition, we have 
provided asymptotic forms of the acceptance rate of HMC under reasonable regularity conditions
and have shown that, as the dimensionality grows to infinity, the forms converge to zero at 
much faster rates than additive TMCMC.
In Section S-8 we also contrasted TMCMC with the transformation-based methods of \ctn{Liu99}, \ctn{Liu00}, and \ctn{Kou05}.

The advantages of TMCMC are more prominent
in high dimensions, where simulating a single random variable can update many parameters
at the same time, thus saving a lot of computing resources. That many variables can be updated
in a single block without compromising
much on the acceptance rate, seems to be another quite substantial advantage provided by our algorithm. 
We illustrated with examples that TMCMC can outperform MH significantly, particularly
in high dimensions. The computational gain of using TMCMC for simulations from doubly intractable
distributions, is also significant, and is illustrated with an example.

In this article we have developed TMCMC for continuous state spaces. 
However, in Section S-9 we show how TMCMC can be generalized to discrete state spaces as well.
A complete development of TMCMC for discrete cases will be a subject of our future work.

\section*{Acknowledgment}
We sincerely thank the reviewer whose comments have led to an improved version of our manuscript.
Conversations with Dr. Ranjan Maitra has also led to improved presentation of some of the ideas.


\end{document}


\bibliographystyle{natbib}

\def\spacingset#1{\renewcommand{\baselinestretch}%
{#1}\small\normalsize} \spacingset{1}


\if0\blind
{
  \title{\bf Supplement to ``Markov Chain Monte Carlo Based on Deterministic Transformations"}
   \author{Somak Dutta\thanks{Corresponding e-mail: sdutta@galton.uchicago.edu}\\
    Department of Statistics\\
    University of Chicago\\
    and\\
   Sourabh Bhattacharya\\
   Bayesian and Interdisciplinary Research Unit\\
  Indian Statistical Institute
 }
\maketitle
 } \fi

\if1\blind
{
 \bigskip
 \bigskip
 \bigskip
 \begin{center}
 {\LARGE\bf Markov Chain Monte Carlo Based on Deterministic Transformations}
 \end{center}

  \medskip
  } \fi

  \bigskip

\renewcommand\thefigure{S-\arabic{figure}}
\renewcommand\thetable{S-\arabic{table}}
\renewcommand\thesection{S-\arabic{section}}

\spacingset{1.45}  

Throughout, we refer to our main article \ctn{Dutta13} as DB. 

\section{Proof of detailed balance for TMCMC}
\label{sec:detailed_balance}
The detailed balance condition is proved as follows: Suppose $\by = T_{\bz}(\bx,\be) \in T_{\bz}(\bx,\Y)$, 
then $\bx = T_{\bz^c}(\by,\be)$. Hence, the kernel $K$ satisfies,
\begin{eqnarray*}
 \pi(\bx)K(\bx \to \by) & = & \pi(\bx)~P(T_{\bz})~g(\be)\min\left\{1,\dfrac{P(T_{\bz^c}) \pi(\by)}
 {P(T_{\bz})\pi(\bx)}J_{\bz}(\bx,\be)\right\} \\
& = & g(\be)~\min\left\{\pi(\bx)~P(T_{\bz}),\pi(\by)P(T_{\bz^c})~J_{\bz}(\bx,\be)\right\}
\end{eqnarray*}
and
\begin{eqnarray*}
 \pi(\by)K(\by \to \bx) & = & \pi(\by)~P(T_{\bz^c})~g(\be)J_{\bz}(\bx,\be)\min\left\{1,\dfrac{P(T_{\bz})
 \pi(\bx)}{P(T_{\bz^c})\pi(\by)}J_{\bz^c}(\by,\be)~\right\} \\
& = & g(\be)~\min\left\{\pi(\by)~P(T_{\bz^c})J_{\bz}(\bx,\be),\pi(\bx)P(T_{\bz})\right\}
\end{eqnarray*}
where $J_{\bz}(\bx,\be) = \left| \partial (T_{\bz}(\bx,\be),\be)/\partial (\bx,\be)\right|$ satisfies
\[ J_{\bz^c}(T_{\bz}(\bx,\be),\be) \times J_{\bz}(\bx,\be) = 1 \quad\textrm{ since }\quad  
T_{\bz^c}(T_{\bz}(\bx,\be),\be) = \bx.  \]

\section{General TMCMC algorithm based on a single $\e$}
\label{sec:general_tmcmc}

\begin{algo}\label{algo:GTM:multivar2} \topline General TMCMC algorithm based on a single $\e$.
\botline \normalfont \ttfamily
\begin{itemize}
 \item Input: Initial value $\supr{\bm x}{0}$, and number of iterations $N$. 
 \item For $t=0,\ldots,N-1$
\begin{enumerate}
 \item Generate $\e \sim g(\cdot)$ and an index $i \sim \mathcal M(1;p_1,\ldots,p_{3^k-1})$ independently.
 Again, actual simulation from the high-dimensional multinomial distribution is not necessary; 
 see Section 3.1 of DB. 
 \item \[ \bm x' = T_{\bz_i}(\supr{\bm x}{t}, \e) \quad \textrm{ and } 
 \quad \alpha(\supr{\bm x}{t}, \e) = \min\left(1, \dfrac{P(T_{\bz^c_i})}{P(T_{\bz_i})} ~\dfrac{\pi(\bm x')}{\pi(\supr{\bm x}{t})} 
 ~\left|\frac{\partial (T_{\bz_i}(\supr{\bm x}{t}, \e),\e)}{\partial(\supr{\bm x}{t}, \e)}\right| \right)\]
\item Set \[ \supr{\bm x}{t+1}= \left\{\begin{array}{ccc}
 \bm x' & \textsf{ with probability } & \a(\supr{\bm x}{t},\e) \\
 \supr{\bm x}{t}& \textsf{ with probability } & 1 - \a(\supr{\bm x}{t},\e)
\end{array}\right.\]
\end{enumerate}
\item End for
\end{itemize}
\botline \rmfamily
\end{algo}

\section{Convergence properties of additive TMCMC}\label{sec:theorems}
In this section we prove some convergence properties of the TMCMC in the case of the additive transformation. 
%
Before going into our main result we first borrow some definitions from the MCMC literature. 
\begin{definition}[Irreducibility] A Markov transition kernel $K$ is $\varphi-$irreducible, where $\varphi$ is a nontrivial measure, if for every $x \in \statesp$ and for every measurable set $A$ of $\statesp$ with $\varphi(A) > 0$, 
there exists $n\in \mathbb N$, such that $K^n(x,A) > 0.$
\end{definition}
\begin{definition}[Small set] A measurable subset $E$ of $\statesp$ is said to be small if there is an $n \in \mathbb N$, a constant $c > 0$, possibly depending on $E$ and a finite measure $\nu$ such that
\[ K^n(x,A) \geq c~\nu(A), \qquad \forall ~A \in \mathcal{B}(\statesp),~\forall ~x \in E\]
\end{definition}

\begin{definition}[Aperiodicity] A Markov kernel $K$ is said to be periodic with period $d > 0$ if the state-space $\statesp$ can be partitioned into $d$ disjoint subsets $\statesp_1,\statesp_2,\ldots,\statesp_d$ with 
\[ K(x,\statesp_{i+1}) = 1 ~\forall ~x\in \statesp_i, ~i=1,2,\ldots,d-1\]
and $K(x,\statesp_1) = 1~\forall ~x \in \statesp_d$.

A Markov kernel $K$ is aperiodic if for no $d\in\mathbb N$ it is periodic with period $d$.
\end{definition}

\subsection{Additive transformation with singleton $\e$}
Consider now the case where $\statesp = \R^k$, $\D = \R$ and 
$T_{\bz}(\mathbf x,\e) = (x_1+z_1a_1\e,x_2+z_2a_2\e,\ldots,x_k+z_ka_k\e)$ 
where, for $i=1,\ldots,k$, $z_i=\pm 1$, and $a_i>0$. In this case $\Y = [0,\infty)$. 
Suppose that $g$ is a density on $\Y$.

\begin{theorem}
 Suppose that $\pi$ is bounded and positive on every compact subset of $\R^k$ and that $g$ is 
 positive on every compact subset of $(0,\infty)$. Then the chain is $\l$-irreducible, aperiodic. 
 Moreover every nonempty compact subset of $\R^k$ is small.
\end{theorem}

\begin{proof}
Without loss we may assume that $a_i=1;$ $i=1,\ldots,k$. 
For notational convenience we shall prove the theorem for $k = 2$. The general case can be 
seen to hold with suitably defined `rotational' matrices on $\R^k$ similar to \eqref{formula:rotmats1}.

Suppose $E$ is a nonempty compact subset of $\R^k$. Let $C$ be a compact rectangle whose sides are parallel to the diagonals $\{(x,y) ~:~ |y| = |x|\}$ and containing $E$ such that $\l(C) > 0$. We shall show that $E$ is small, i.e., $\exists ~c > 0$ such that 
\[ K^2(\bx,A) \geq c \l_C(A) \qquad \forall A \in \mathcal B(\R^2) \textrm{ and } \forall x \in E.\]
It is clear that the points reachable from $\bx$ \emph{in two steps} are of the form
\[\begin{pmatrix} x_1 \pm \e_1 \pm \e_2 \\ x_2 \pm \e_1 \pm \e_2\end{pmatrix},\qquad \e_1 \geq 0, \e_2 \geq 0 \]
Thus, if we define the matrices
\begin{equation}\label{formula:rotmats1}
\begin{split}
 & M_1 = \begin{pmatrix} 1 & 1 \\ 1 & -1 \end{pmatrix} \quad M_2 = \begin{pmatrix} -1 & 1 \\ 1 & 1 \end{pmatrix} \quad M_3 = \begin{pmatrix} 1 & -1 \\ -1 & -1 \end{pmatrix} \quad M_4 = \begin{pmatrix} -1 & -1 \\ -1 & 1 \end{pmatrix} \\
& \tilde M_1 = \begin{pmatrix} 1 & 1 \\ -1 & 1 \end{pmatrix} \quad \tilde M_2 = \begin{pmatrix} 1 & -1 \\ 1 & 1 \end{pmatrix} \quad \tilde M_3 = \begin{pmatrix} -1 & 1 \\ -1 & -1 \end{pmatrix} \quad \tilde M_4 = \begin{pmatrix} -1 & -1 \\ 1 & -1 \end{pmatrix}
\end{split}
\end{equation}
then the points reachable from $\bx$ \emph{in two steps}, other than the points lying on the diagonals passing through $\bx$ itself, are of the form
\[ \bx + M_i\left(\begin{smallmatrix} \e_1 \\ \e_2\end{smallmatrix}\right) \quad\textrm{ and }\quad  \bx + \tilde M_i\left(\begin{smallmatrix} \e_1 \\ \e_2\end{smallmatrix}\right), \quad \e_1 > 0, \e_2 > 0,~i=1,\ldots,4.\]

Define
\[ m = \inf_{\by \in C}\pi(\by) > 0 \qquad M = \sup_{\by \in C}\pi(\by) < \infty \qquad a = \inf_{0 < \e < R}g(\e) > 0\]
where $R$ is the length of the diagonal of the rectangle $C$\footnote{Actually $R/\sqrt{2}$ suffices.}.
Fix an element $\bx \in E$. For any set $A \in \mathcal B(\R^2)$, let $A^* = A\cap C$ and define,
\begin{equation}
\begin{split}
A_i & = \{ \be \in (0,\infty)^2 ~:~ \bx + M_i \be\in A^* \} \\
\tilde A_i & = \{ \be \in (0,\infty)^2 ~:~ \bx + \tilde M_i \be\in A^* \}
\end{split}
\end{equation}
The need for defining such sets illustrated in the following example: to make a transition from the state $\bx$ to a state in $A^*$ in two steps, first making a forward transition in both coordinates and then a forward transition in first coordinate and a backward transition in the second coordinate is same as applying the transformation $\bx \to \bx + M_1\be$ for some $\be \in A_1$ in two steps, i.e. first 
\[\bx \to \bx + M_1(\e_1,0)^T = \bx + (\e_1,\e_1)^T \quad\textrm{ then } \quad \bx + M_1(\e_1,\e_2)^T \to \bx + M_1\be\]
Also note that for any $\be = (\e_1,\e_2)\in A_i$, $A^* \subset C$ implies that the intermediate point $\bx + M_i(\e_1,0)^T \in C$ and similarly for $\tilde A_i~(i=1,\ldots,4)$.
Now, with ${\underline p}$ and $\bar p$ as the minimum and maximum of the move probabilities. 
\begin{eqnarray}\label{eqn:smallsetadditive}
 & & K^2(\bx,A) \geq K^2(\bx,A^*) \nonumber \\
& \geq & {\underline p}^2\sum_{i=1}^4 \dint_{A_i} g(\e_1)g(\e_2)\min\left\{ \frac{{\underline p}\pi(\bx + M_i(\e_1,0)^T)}{{\bar p}\pi(\bx)},1\right\}\min\left\{ \frac{{\underline p}\pi(\bx + M_i(\e_1,\e_2)^T)}{{\bar p}\pi(\bx + M_i(\e_1,0)^T)},1\right\}d\e_1d\e_2 \nonumber\\
& + & {\underline p}^2\sum_{i=1}^4 \dint_{\tilde A_i} g(\e_1)g(\e_2)\min\left\{ \frac{{\underline p}\pi(\bx + \tilde M_i(\e_1,0)^T)}{{\bar p}\pi(\bx)},1\right\}\min\left\{ \frac{{\underline p}\pi(\bx + \tilde M_i(\e_1,\e_2)^T)}{{\bar p}\pi(\bx + \tilde M_i(\e_1,0)^T)},1\right\}d\e_1d\e_2 \nonumber \\
& \geq & {\underline p}^2 a^2\left(\min\left\{\dfrac{\underline p m}{\bar p M},1\right\}\right)^2\left(\sum_{i=1}^4\l(A_i) + \sum_{i=1}^4\l(\tilde A_i) \right) \nonumber \\
& = & {\underline p}^2 a^2\left(\min\left\{\dfrac{\underline p m}{\bar p M},1\right\}\right)^2 \times 2 \times \sum_{i=1}^4\l(A_i)
\end{eqnarray}
Since $(\e_1,\e_2) \in A_i \iff (\e_2,\e_1) \in \tilde A_i$, so that, $\l(A_i) = \l(\tilde A_i)$.
Now notice that, if we define for $i=1,\ldots,4$
\[ f_i : (0,\infty)^2 \to \R^2 \ni \be \mapsto \bx + M_i\be \]
 and 
\[ A_\bx = \{ (\e,0)^T ~:~ \e > 0, (x_1 \pm \e, x_2 \pm \e) \in A^* \}\]
then,
\[A^* = \bigcup_{i=1}^4f_i(A_i\cup A_x) \quad\imply\quad \l(A^*) = \sum_{i=1}^4f_i(A_i) \quad = \quad 2 \times \sum_{i=1}^4\l(A_i),\]
since, $f_i(A_i)$'s are pairwise disjoint, $\l(f_i(A_x)) = 0$ and $\l(f_i(A_i)) = 2\l(A_i)$ for $1\leq i \leq 4$. It follows from \eqref{eqn:smallsetadditive} that
\[ K^2(\bx,A) \quad\geq \quad {\underline p}^2 a^2\left(\min\left\{\dfrac{\underline p m}{\bar p M},1\right\}\right)^2 \l(A^*) \quad = \quad c\l_C(A) \]
where $c = {\underline p}^2 a^2\left(\min\left\{\dfrac{\underline p m}{\bar p M},1\right\}\right)^2 > 0$. \\ This completes the proof that $E$ is small.

That the chain is irreducible, follows easily, for any $\bx$, the set $\{\bx\}$ is a compact set and for a measurable set $A$ with $\l(A) > 0$ we may choose $C$ in the first part of the proof such that $\l(C\cap A) > 0$. Now,
\[ K^2(\bx,A) \quad \geq \quad c\l(C\cap A) > 0\]
Also aperiodicity follows trivially from the observation that any set with positive $\l$-measure can be accessed in at most 2 steps.
\end{proof}


\section{General TMCMC algorithm with single $\e$ and dependent $\bz$}
\label{sec:tmcmc_dependent_z}
Also, let 
Let $h_1(\bp)$, $h_2(\bq)$
be the specified joint distributions of $\bp$ and $\bq$ induced by the Gaussian distributions
of $\bw_1,\bw_2,\bw_3$, and
let $P(\bz\vert \bp,\bq)=\prod_{i=1}^kf_i(z_i\vert p_i,q_i)$ denote the conditional probability
of $\bz$, given $\bp,\bq$, where $f_i(\cdot\vert p_i,q_i)$ is the conditional
probability of $z_i$ given $p_i$ and $q_i$. 
Then the general TMCMC algorithm with singleton $\e$ and dependent $\bz$ is given as follows.

\begin{algo}\label{algo:GTM:multivar3} \topline General TMCMC algorithm based on single $\e$
and dependent $\bz$.
\botline \normalfont \ttfamily
\begin{itemize}
 \item Input: Initial value $\supr{\bm x}{0}$, and number of iterations $N$. 
\begin{enumerate}
 \item For $t=0,\ldots,N-1$
\begin{enumerate}
 \item Generate $\bw_1\sim N_k(\bmu_1,\bSigma_1)$, $\bw_2\sim N_k(\bmu_2,\bSigma_2)$, and $\bw_3\sim N_k(\mu_3,\bSigma_3)$. 
 \item For $i=1,\ldots,k$, set $p_i=\exp\left(w_{1i}\right)/\sum_{j=1}^3\exp\left(w_{ji}\right)$, 
 $q_i=\exp\left(w_{2i}\right)/\sum_{j=1}^3\exp\left(w_{ji}\right)$, and
 $1-p_i-q_i=\exp\left(w_{3i}\right)/\sum_{j=1}^3\exp\left(w_{ji}\right)$.
 \item Generate $\e \sim g(\cdot)$ and an index $i \sim \mathcal M(1;p_1,\ldots,p_{3^k-1})$ independently.
 \end{enumerate}
 \item \[ \bm x' = T_{\bz_i}(\supr{\bm x}{t}, \e) \quad \textrm{ and } 
 \quad \alpha(\supr{\bm x}{t}, \e) = \min\left(1, \dfrac{P(\bz^c_i\vert\bp,\bq)}{P(\bz_i\vert\bp,\bq)} 
 ~\dfrac{\pi(\bm x')}{\pi(\supr{\bm x}{t})} 
 ~\left|\frac{\partial (T_{\bz_i}(\supr{\bm x}{t}, \e),\e)}{\partial(\supr{\bm x}{t}, \e)}\right| \right)\]
\item Set \[ \supr{\bm x}{t+1}= \left\{\begin{array}{ccc}
 \bm x' & \textsf{ with probability } & \a(\supr{\bm x}{t},\e) \\
 \supr{\bm x}{t}& \textsf{ with probability } & 1 - \a(\supr{\bm x}{t},\e)
\end{array}\right.\]
\end{enumerate}
\item End for
\end{itemize}
\botline \rmfamily
\end{algo}

\section{Proof of detailed balance for TMCMC with dependent $\bz$}
\label{sec:detailed_balance_dependent_z}

Let $\by = T_{\bz}(\bx,\e) \in T_{\bz}(\bx,\Y)$, 
then $\bx = T_{\bz^c}(\by,\e)$. 
The kernel $K$ satisfies,
\begin{eqnarray*}
 \pi(\bx)K(\bx \to \by) & = & \pi(\bx)h_1(\bp)h_2(\bq)
 P(\bz\vert\bp,\bq)~g(\e)\min\left\{1,\dfrac{P(\bz^c\vert\bp,\bq) \pi(\by)}
 {P(\bz\vert\bp,\bq)\pi(\bx)}J_{\bz}(\bx,\e)\right\} \\
& = & h_1(\bp)h_2(\bq)g(\be)\min\left\{\pi(\bx)P(\bz\vert\bp,\bq),\pi(\by)P(\bz^c\vert\bp,\bq)
J_{\bz}(\bx,\e)\right\}
\end{eqnarray*}
and
\begin{eqnarray*}
 \pi(\by)K(\by \to \bx) & = & \pi(\by)h_1(\bp)h_2(\bq)P(\bz^c\vert\bp,\bq)g(\e)J_{\bz}(\bx,\e)
 \min\left\{1,\dfrac{P(\bz\vert\bp,\bq)
 \pi(\bx)}{P(\bz^c\vert\bp,\bq)\pi(\by)}J_{\bz^c}(\by,\e)~\right\} \\
& = & h_1(\bp)h_2(\bq)g(\e)\min\left\{\pi(\by)P(\bz^c\vert\bp,\bq)J_{\bz}(\bx,\e),\pi(\bx)P(\bz\vert\bp,\bq)
\right\}
\end{eqnarray*}

\section{Improved acceptance rates of additive TMCMC with singleton $\e$ compared 
to joint updating using RWMH}
\label{sec:accept_rate}

The joint RWMH algorithm generates $\be=(\e_1,\ldots,\e_k)'$ independently from $N(0,1)$,
and then uses the transformation of the form $x'_i=x_i+a_i\e_i$, where $a_i>0$ are
appropriate scaling constants. For large $k$, the so-called ``curse of dimensionality"
can force the acceptance rate to be close to zero. 
On the other hand, the additive-transformation based TMCMC also updates 
$(x_1,\ldots,x_k)$ simultaneously in a single block, but instead of using $k$ different $\e_i$,
it uses a single $\e$ for updating all the $x_i$ variables. In other words, for TMCMC
based on additive transformation $\be$ is of the form 
$\be=(\pm\e,\ldots,\pm\e)'$, where $\e\sim N(0,1)\mathbb I_{\{\e>0\}}$.
Thus, relative to RWMH, the dimension in the TMCMC case is effectively reduced to 1, 
avoiding the curse of dimensionality. Thus, it is expected that additive TMCMC will
have a much higher acceptance rate than RWMH.
In this section we formalize and compare the issues related to acceptance rates of additive TMCMC
and RWMH.

\subsection{Discussion on optimal scaling and optimal acceptance rate of additive TMCMC and RWMH}
\label{subsec:optimal_scaling}

A reasonable approach to compare the acceptance rates of additive TMCMC and RWMH is to 
develop the optimal scaling theory for additive TMCMC, obtain
the optimal acceptance rate, and then compare the latter with the optimal acceptance rates for RWMH,
which are already established in the MCMC literature. 
Indeed, optimal scaling and optimal acceptance rate of additive TMCMC and comparison with
those of RWMH is the subject of \ctn{Dey13b}, where it is shown that additive TMCMC has
a much higher optimal acceptance rate compared to RWMH. Before we summarize the results
of \ctn{Dey13b} we first provide a brief overview of optimal scaling and optimal acceptance rate of RWMH.

\subsubsection{Brief overview of optimal scaling and optimal acceptance rate for RWMH}  
\label{subsubsec:optimal_rwmh}

Roughly, the optimal random walk proposal variance, represented as an inverse function of the dimension $k$,
is the one that maximizes the speed of 
convergence to the stationary distribution of a relevant
diffusion process to which a `sped-up' version of RWMH weakly converges as the dimension $k$
increases to infinity. The optimal acceptance rate corresponds to the optimal proposal variance.
Under various assumptions on the form of the target distribution $\pi$, ranging from the
$iid$ assumption (\ctn{Roberts97}), through independent but non-identical set-up (\ctn{Bedard07}),
to a more general dependent structure (\ctn{Mattingly11}), the optimal acceptance rate
turns out to be 0.234. 

\subsubsection{Optimal scaling and optimal acceptance rate for additive TMCMC}
\label{subsubsec:optimal_tmcmc}

In \ctn{Dey13b} it has been proved
in the case of additive TMCMC, assuming $p_i=q_i=1/2$, 
that the optimal acceptance rate, as $k\rightarrow\infty$, is 0.439 
under the set-ups ($iid$, independent but non-identical, and dependent) for which the optimal 
acceptance rate for RWMH has been studied and established to be 0.234. 
Thus, the optimal acceptance rate for additive TMCMC is much higher than that of RWMH. 
The optimal scalings, that is, the optimal values of the scales $a_1,\ldots,a_k$ are also available
using the optimal scaling theory. 
As shown in \ctn{Dey13b}, all these results for additive TMCMC and RWMH 
remain true even in all the aforementioned set-ups if some of the co-ordinates of $\bx$ are
updated at random, conditioning on the remaining co-ordinates.   

\subsection{Comparison between the asymptotic forms of the acceptance rates of additive 
TMCMC and RWMH for strongly log-concave target densities}
\label{subsec:lower_upper_bound}
The results on optimal scaling and optimal acceptance rate discussed in Sections \ref{subsubsec:optimal_rwmh}
and \ref{subsubsec:optimal_tmcmc} are available only for special forms of the target distribution $\pi$. 
In this section we obtain the asymptotic forms of the acceptance rates associated with RWMH and additive TMCMC
assuming that the target density is strongly log-concave. 
In particular, under suitable conditions 
we show that as the dimension increases,
the acceptance rate of RWMH converges to zero at a much faster rate than that of additive TMCMC.

Assuming without loss of generality that the marginal variances of the target density $\pi$ are all unity 
(achieved after suitable scaling perhaps), for RWMH we consider the following proposed value $\bx'$ given
the current value $\bx$: $\bx'=\bx+\be$,
where $\be\sim N_k(\bzero,\bI_k)$. On the other hand, for additive TMCMC, we consider $\bx'=\bx+\e\bdelta$,
where $\epsilon \sim N(0,1)\mathbb{I}(\epsilon > 0)$ and the components $\delta_i$ 
of $\boldsymbol\delta$ are $iid$ taking values $\pm 1$ with probability $1/2$ each. 

To proceed we consider the following form of acceptance rate for our asymptotic framework.
Letting $R(\bx'|\bx)$ denote the acceptance probability of $\bx'$ given the current value $\bx$,
and letting $U\sim Uniform(0,1)$, the acceptance rate is given by
\begin{align}
AR &= \int R(\bx'|\bx)q(\bx'|\bx)\pi(\bx)d\bx d\bx'\notag
\\[1ex]
&=\int Pr\left(U<R(\bx'|\bx)\right)q(\bx'|\bx)\pi(\bx)d\bx d\bx'\notag
\\[1ex]
&=\int \left[\int Pr\left(U<R(\bx'|\bx)\right)q(\bx'|\bx)d\bx'\right]\pi(\bx)d\bx\notag
\\[1ex]
&=\int \left[\int_0^1 Pr\left(R(\bx'|\bx)>u\right)du\right]\pi(\bx)d\bx
\label{eq:ineq4}
\end{align}
In the above formula for acceptance rate note that $Pr\left(R(\bx'|\bx)>u\right)\rightarrow 1$ as $u\rightarrow 0$
and $Pr\left(R(\bx'|\bx)>u\right)\rightarrow 0$ as $u\rightarrow 1$.
Hence, given any $\eta_1>0,\eta_2>0$, we can choose $\psi_1,\psi_2\in (0,1)$ sufficiently small such that 
$\int_0^{\psi_1}Pr\left(R(\bx'|\bx)>u\right)du<\eta_1$ and $\int_{1-\psi_2}^1 Pr\left(R(\bx'|\bx)>u\right)du<\eta_2$.
Hence, re-writing (\ref{eq:ineq4})
as
\begin{eqnarray}
AR&=&\int \left[\int_0^{\psi_1} Pr\left(R(\bx'|\bx)>u\right)du\right]\pi(\bx)d\bx
+\int \left[\int_{\psi_1}^{1-\psi_2} Pr\left(R(\bx'|\bx)>u\right)du\right]\pi(\bx)d\bx\nonumber\\
&&\quad\quad\quad\quad+\int \left[\int_{1-\psi_2}^{1} Pr\left(R(\bx'|\bx)>u\right)du\right]\pi(\bx)d\bx,\nonumber
\end{eqnarray}
we find that the first and the third term on the right hand side are negligible for any algorithm. So,
for the purpose of comparing algorithms with respect to their acceptance rates, we consider only the middle term; 
in all that follow we denote 
\begin{equation}
AR = \int \left[\int_{\psi_1}^{1-\psi_2} Pr\left(R(\bx'|\bx)>u\right)du\right]\pi(\bx)d\bx.
\label{eq:ineq5}
\end{equation}

For our purpose, we consider a target density $\pi(\mathbf{x})$ of $k$ variables that is strongly log-concave, 
that is,
\begin{equation}\label{eqn:stronglyconvex}
 -M_k \mathbf{I}_k \leq \nabla^2 \log \pi(\mathbf{x}) \leq -m_k \mathbf{I}_k,
\end{equation}
where we assume that $M_k=c_k+m_k$, with $m_k,~c_k>0$ for every $k$. We further assume that 
$m_k\rightarrow\infty$, and the sequence $\{c_k\}$ is such that $c_k/m_k \rightarrow 0$ as $k\rightarrow\infty$. 
Then clearly, $M_k\asymp m_k$, meaning $M_k/m_k\rightarrow 1$ as $k\rightarrow\infty$. In fact, we assume that
$M_k/m_k$ approaches 1 at a sufficiently fast rate, so that
$k\left\vert\frac{M_k}{m_k}-1\right\vert\rightarrow 0$. For our purpose we assume that 
$c_k=O(k^s); s\geq 1$ and 
$m_k=O(k^t);t>s+1\geq 2$,
so that $M_k=O(k^t)$. It is easy to verify that these choices satisfy the above conditions.

It is important to note that our assumption $m_k,M_k\rightarrow\infty$ need not hold for all
strongly log-concave distributions. For instance,
when $\pi$ is the $iid$ product of standard normals, that is, when 
$\bx\sim N_k\left(\bzero,\bI_k\right)$ under $\pi$, $\nabla^2 \log \pi(\mathbf{x})=\bI_k$.
In this case, $m_k=M_k=1$ for every $k\geq 1$. In general, even if $m_k$ and $M_k$ remains finite
as $k$ grows to infinity, our proofs remain valid provided that $M_k\asymp m_k$ and 
$k\left\vert\frac{M_k}{m_k}-1\right\vert\rightarrow 0$. The case of $\pi$ being an $iid$ product
of standard normals clearly satisfies the above conditions.

\subsubsection{Asymptotic form of the acceptance rate for RWMH}
\label{subsubsec:asymp_rwmh}

Let $\bx^*$ denote the mode of the target density $\pi(\cdot)$.
Then for every $r \in (0,1),$
\begin{eqnarray*}
& & Pr~(R(\mathbf{x}'| \mathbf{x}) < r)  = Pr~(\pi(\mathbf{x}')/\pi(\mathbf{x}) < r) = Pr~(\log \pi(\mathbf{x}') - \log \pi(\mathbf{x})  < \log r )\\
& = &  Pr~(\left[\log \pi(\mathbf{x}') - \log \pi(\mathbf{x^*})\right]-\left[\log \pi(\mathbf{x}') - \log \pi(\mathbf{x^*})\right]  < \log r )\\
& = & Pr~\left(\left[\nabla \log \pi(\mathbf{x^*})^T (\mathbf{x}' - \mathbf{x^*}) + (1/2)(\mathbf{x}'-\mathbf{x^*})^T \nabla^2 \log \pi(\boldsymbol{\xi_1(\mathbf x',\mathbf x^*)})(\mathbf{x}'-\mathbf{x^*})\right]\right.\\
&&\quad\quad\left.-\left[\nabla \log \pi(\mathbf{x^*})^T (\mathbf{x} - \mathbf{x^*}) + (1/2)(\mathbf{x}-\mathbf{x^*})^T \nabla^2 \log \pi(\boldsymbol{\xi_2(\mathbf x,\mathbf x^*)})(\mathbf{x}-\mathbf{x^*})\right]
< \log r\right), \\
&&\quad\textrm{~for some } \boldsymbol\xi_1(\mathbf x',\mathbf x^*), \boldsymbol\xi_2(\mathbf x,\mathbf x^*)\ \ \mbox{depending upon}\ \ (\mathbf x',\mathbf x^*) \ \ \mbox{and} (\mathbf x,\mathbf x^*) \ \ \mbox{respectively};\\
& = & Pr~\left(\left[(1/2)(\mathbf{x}'-\mathbf{x^*})^T \nabla^2 \log \pi(\boldsymbol{\xi_1(\mathbf x',\mathbf x^*)})(\mathbf{x}'-\mathbf{x^*})\right]\right.\\
&&\quad\quad\left. -\left[(1/2)(\mathbf{x}-\mathbf{x^*})^T \nabla^2 \log \pi(\boldsymbol{\xi_2(\mathbf x',\mathbf x^*)})(\mathbf{x}-\mathbf{x^*})\right]<\log r\right)\\
&&\quad\quad\quad\quad\quad\mbox{since}\ \ \nabla \log \pi(\mathbf{x^*})=\bzero.
\end{eqnarray*}
Thus from the assumption in \eqref{eqn:stronglyconvex}, and noting that
$(\mathbf{x}'-\mathbf{x^*})^T(\mathbf{x}'-\mathbf{x^*})=(\mathbf{x}-\mathbf{x^*})^T(\mathbf{x}-\mathbf{x^*})
+2(\mathbf{x}-\mathbf{x^*})^T\boldsymbol\epsilon+\boldsymbol\epsilon^T\boldsymbol\epsilon$
it follows that 
\begin{equation}\label{eqn:bound1}
 \begin{split}
  & Pr~\left(\frac{(M_k-m_k)}{2}(\mathbf{x}-\mathbf{x^*})^T (\mathbf{x}-\mathbf{x^*}) 
  -m_k(\mathbf{x}-\mathbf{x^*})^T\boldsymbol\epsilon - \frac{m_k}{2}\boldsymbol\epsilon^T\boldsymbol\epsilon
  < \log r\right) \\
 & \leq Pr~(R(\mathbf{x}'| \mathbf{x}) < r) \\
 & \leq 
   Pr~\left(-\frac{(M_k-m_k)}{2}(\mathbf{x}-\mathbf{x^*})^T (\mathbf{x}-\mathbf{x^*}) 
  -M_k(\mathbf{x}-\mathbf{x^*})^T\boldsymbol\epsilon - \frac{M_k}{2}\boldsymbol\epsilon^T\boldsymbol\epsilon
  < \log r\right);
 \end{split}
\end{equation}
so that
\begin{equation}\label{eqn:new_bound2}
 \begin{split}
  & Pr~\left(\frac{(M_k-m_k)}{2}(\mathbf{x}-\mathbf{x^*})^T (\mathbf{x}-\mathbf{x^*}) 
  -m_k(\mathbf{x}-\mathbf{x^*})^T\boldsymbol\epsilon - \frac{m_k}{2}\boldsymbol\epsilon^T\boldsymbol\epsilon
  > \log r\right) \\
 & \geq Pr~(R(\mathbf{x}'| \mathbf{x}) > r) \\
 & \geq 
   Pr~\left(-\frac{(M_k-m_k)}{2}(\mathbf{x}-\mathbf{x^*})^T (\mathbf{x}-\mathbf{x^*}) 
  -M_k(\mathbf{x}-\mathbf{x^*})^T\boldsymbol\epsilon - \frac{M_k}{2}\boldsymbol\epsilon^T\boldsymbol\epsilon
  > \log r\right),
 \end{split}
\end{equation}
Hence, using (\ref{eq:ineq5}) it can be seen that the acceptance rate is bounded above and below as follows
\begin{align}
& \int \left[\int_{\psi_1}^{1-\psi_2} \left\{ \int_{A^k_{2,\be,u}}\frac{1}{(2\pi)^{k/2}}\exp\left\{-\frac{1}{2}\be^T\be\right\}d\be \right\}du\right]\pi(\bx)d\bx\notag
\\[1ex]
&\geq AR^{(RWMH)}\label{eq:ar_rwmh} 
\\[1ex]
&\geq \int \left[\int_{\psi_1}^{1-\psi_2} \left\{ \int_{A^k_{1,\be,u}}\frac{1}{(2\pi)^{k/2}}\exp\left\{-\frac{1}{2}\be^T\be\right\}d\be \right\}du\right]\pi(\bx)d\bx,
\notag
\end{align}
where
$$A^k_{1,\be,u}=\left\{\mathbf{x}: -\frac{(M_k-m_k)}{2}(\mathbf{x}-\mathbf{x^*})^T (\mathbf{x}-\mathbf{x^*})
-M_k(\mathbf{x}-\mathbf{x^*})^T\boldsymbol\epsilon - \frac{M_k}{2}\boldsymbol\epsilon^T\boldsymbol\epsilon
> \log u\right\}$$ and
$$A^k_{2,\be,u}=\left\{\mathbf{x}: \frac{(M_k-m_k)}{2}(\mathbf{x}-\mathbf{x^*})^T (\mathbf{x}-\mathbf{x^*})
-m_k(\mathbf{x}-\mathbf{x^*})^T\boldsymbol\epsilon - \frac{m_k}{2}\boldsymbol\epsilon^T\boldsymbol\epsilon
> \log u\right\}.$$
Now, note that for some $\bxi(\bx,\bx^*)$ depending upon $\bx$ and $\bx^*$,
\begin{eqnarray}
\pi(\bx)
 &=&\exp\left\{\log\pi(\mathbf{x})\right\}d\mathbf{x}\nonumber\\
 &=& \exp\left\{\log\pi(\mathbf{x^*})
 +\frac{1}{2}(\mathbf{x}-\mathbf{x^*})^T\nabla^2\log\pi(\boldsymbol\xi(\mathbf x,\mathbf x^*))
 (\mathbf{x}-\mathbf{x^*})\right\}, \nonumber\\
\end{eqnarray}
so that the inequalities related to strong convexity, given by (\ref{eqn:stronglyconvex}) yield
\begin{equation*}
\frac{(2\pi)^{k/2}}{m^k_k}\pi(\mathbf x^*)
\frac{m^k_k}{(2\pi)^{k/2}} \exp\left\{-\frac{m_k}{2}(\mathbf{x}-\mathbf{x^*})^T(\mathbf{x}-\mathbf{x^*})\right\}
\geq \pi(\bx)
\end{equation*}
\begin{equation}
\geq \frac{(2\pi)^{k/2}}{M^k_k}\pi(\mathbf x^*)
\frac{M^k_k}{(2\pi)^{k/2}} \exp\left\{-\frac{M_k}{2}(\mathbf{x}-\mathbf{x^*})^T(\mathbf{x}-\mathbf{x^*})\right\}
\label{eq:eq2}
\end{equation}

Using the lower bound of $\pi(\bx)$ given by (\ref{eq:eq2}) and Fubini's theorem, the lower bound
of the acceptance rate given by (\ref{eq:ar_rwmh}) can be further bounded below as
\begin{align}
AR^{(RWMH)} &\geq \int \int_{\psi_1}^{1-\psi_2}  \int_{A^k_{1,\be,u}}\frac{1}{(2\pi)^{k/2}}\exp\left\{-\frac{1}{2}\be^T\be\right\}\pi(\bx)~d\bx~du~d\be \notag
\\[1ex]
&\geq \frac{(2\pi)^{k/2}}{M^k_k}\pi(\mathbf x^*)\int \int_{\psi_1}^{1-\psi_2}  \int_{A^k_{1,\be,u}}\frac{1}{(2\pi)^{k/2}}\exp\left\{-\frac{1}{2}\be^T\be\right\}\notag
\\[1ex]
&\quad\quad\quad\quad\times\frac{M^k_k}{(2\pi)^{k/2}} \exp\left\{-\frac{M_k}{2}(\mathbf{x}-\mathbf{x^*})^T(\mathbf{x}-\mathbf{x^*})\right\}~d\bx~du~d\be\notag
\\[1ex]
&\geq \frac{(2\pi)^{k/2}}{M^k_k}\pi(\mathbf x^*)\inf_{u\in (\psi_1,1-\psi_2)}Pr~(A^k_{1,\be,u}),
\label{eq:lowerbound1}
\end{align}
where $Pr~(A^k_{1,\be,u})$ must be calculated with respect to $\be\sim N_k(\bzero,\bI_k)$, and
independently, $\bx-\bx^*\sim N_k(\bzero,M^{-1}_k\bI_k)$. 

Similarly, using the upper bound of $\pi(\bx)$ given by (\ref{eq:eq2}) the upper bound
of the acceptance rate given by (\ref{eq:ar_rwmh}) can be further bounded above as
\begin{align}
AR^{(RWMH)} &\leq \int \int_{\psi_1}^{1-\psi_2}  \int_{A^k_{2,\be,u}}\frac{1}{(2\pi)^{k/2}}\exp\left\{-\frac{1}{2}\be^T\be\right\}\pi(\bx)~d\bx~du~d\be \notag
\\[1ex]
&= \frac{(2\pi)^{k/2}}{m^k_k}\pi(\mathbf x^*)\int_{\psi_1}^{1-\psi_2} Pr~(A^k_{2,\be,u})du\notag
\\[1ex]
&\leq \frac{(2\pi)^{k/2}}{m^k_k}\pi(\mathbf x^*)\sup_{u\in (\psi_1,1-\psi_2)}Pr~(A^k_{2,\be,u})
\label{eq:upperbound1}
\end{align}
The probability 
$Pr~(A^k_{2,\be,u_k})$ must be calculated with respect to $\be\sim N_k(\bzero,\bI_k)$, and
independently, $\bx-\bx^*\sim N_k(\bzero,m^{-1}_k\bI_k)$. 
Thus, we have
\begin{equation}
\frac{(2\pi)^{k/2}}{M^k_k}\pi(\mathbf x^*)\inf_{u\in (\psi_1,1-\psi_2)}Pr~(A^k_{1,\be,u})
\leq AR^{(RWMH)}
\leq \frac{(2\pi)^{k/2}}{m^k_k}\pi(\mathbf x^*)\sup_{u\in (\psi_1,1-\psi_2)}Pr~(A^k_{2,\be,u}).
\label{eq:lower_upper1}
\end{equation}

We first focus on the lower bound in (\ref{eq:lower_upper1}).
As $k\rightarrow\infty$,
\\[2mm]
$-\frac{(M_k-m_k)}{2}(\mathbf{x}-\mathbf{x^*})^T (\mathbf{x}-\mathbf{x^*})
-M_k(\mathbf{x}-\mathbf{x^*})^T\boldsymbol\epsilon - \frac{M_k}{2}\boldsymbol\epsilon^T\boldsymbol\epsilon$
\begin{eqnarray}
&\sim &
AN\left(-\frac{k}{2}\left[\left(\frac{M_k-m_k}{M_k}\right)+M_k\right],
\frac{k}{2}\left[\left(\frac{M_k-m_k}{M_k}\right)^2+2M_k+M^2_k\right]\right),
\label{eq:an1}
\end{eqnarray}
where $AN(\mu,\sigma^2)$ denotes 
asymptotic normal with mean $\mu$ and variance $\sigma^2$. From (\ref{eq:an1}) it follows that
\begin{eqnarray}
\inf_{u\in (\psi_1,1-\psi_2)}Pr~(A^k_{1,\be,u}) &\asymp &  
1-\sup_{u\in (\psi_1,1-\psi_2)}
\Phi\left(\frac{\log u+\frac{k}{2}\left[\left(\frac{M_k-m_k}{M_k}\right)+M_k\right]}
{\sqrt{\frac{k}{2}\left[\left(\frac{M_k-m_k}{M_k}\right)^2+2M_k+M^2_k\right]}}\right)\nonumber\\
&=& 1- \Phi\left(\frac{\log (1-\psi_2)+\frac{k}{2}\left[\left(\frac{M_k-m_k}{M_k}\right)+M_k\right]}
{\sqrt{\frac{k}{2}\left[\left(\frac{M_k-m_k}{M_k}\right)^2+2M_k+M^2_k\right]}}\right).
\label{eq:infimum1}
\end{eqnarray}
Combining (\ref{eq:lowerbound1}) and (\ref{eq:infimum1}) we obtain
\begin{equation}
AR^{(RWMH)}\geq
\frac{(2\pi)^{k/2}}{M^k_k}\pi(\mathbf x^*)
\left\{1-\Phi\left(\frac{\log (1-\psi_2)+\frac{k}{2}\left[\left(\frac{M_k-m_k}{M_k}\right)+M_k\right]}
{\sqrt{\frac{k}{2}\left[\left(\frac{M_k-m_k}{M_k}\right)^2+2M_k+M^2_k\right]}}\right)
\right\}.
\label{eq:AR_RWMH_lowerbound1}
\end{equation}

Now focusing our attention on the upper bound of $AR^{(RWMH)}$ we similarly obtain
\begin{eqnarray}
AR^{(RWMH)}
&\leq &\frac{(2\pi)^{k/2}}{m^k_k}\pi(\mathbf x^*)
\left\{1-
\Phi\left(\frac{\log \psi_1-\frac{k}{2}\left[\left(\frac{M_k-m_k}{m_k}\right)-m_k\right]}
{\sqrt{\frac{k}{2}\left[\left(\frac{M_k-m_k}{m_k}\right)^2+2m_k+m^2_k\right]}}\right)\right\}.\nonumber\\
\label{eq:AR_RWMH_upperbound1}
\end{eqnarray}
In other words,
\begin{align}\label{eq:AR_RWMH_bounds}
\frac{(2\pi)^{k/2}}{M^k_k}\pi(\mathbf x^*)
&\left\{1-\Phi\left(\frac{\log (1-\psi_2)+\frac{k}{2}\left[\left(\frac{M_k-m_k}{M_k}\right)+M_k\right]}
{\sqrt{\frac{k}{2}\left[\left(\frac{M_k-m_k}{M_k}\right)^2+2M_k+M^2_k\right]}}\right)
\right\}
\leq  AR^{(RWMH)}\notag\\
&\leq \frac{(2\pi)^{k/2}}{m^k_k}\pi(\mathbf x^*)
\left\{1-\Phi\left(\frac{\log \psi_1-\frac{k}{2}\left[\left(\frac{M_k-m_k}{m_k}\right)-m_k\right]}
{\sqrt{\frac{k}{2}\left[\left(\frac{M_k-m_k}{m_k}\right)^2+2m_k+m^2_k\right]}}\right)\right\}.\nonumber\\
\end{align}
Since $m_k\asymp M_k$, it is easy to see that 
\begin{align}\label{eq:AR_RWMH_bounds2}
&\frac{\log (1-\psi_2)+\frac{k}{2}\left[\left(\frac{M_k-m_k}{M_k}\right)+M_k\right]}
{\sqrt{\frac{k}{2}\left[\left(\frac{M_k-m_k}{M_k}\right)^2+2M_k+M^2_k\right]}}
\asymp  \sqrt{\frac{k}{2}},\quad\mbox{and}\notag\\
&\frac{\log \psi_1-\frac{k}{2}\left[\left(\frac{M_k-m_k}{m_k}\right)-m_k\right]}
{\sqrt{\frac{k}{2}\left[\left(\frac{M_k-m_k}{m_k}\right)^2+2m_k+m^2_k\right]}}
\asymp  \sqrt{\frac{k}{2}}.\notag
\end{align}
Hence, it follows that
\begin{equation}
AR^{(RWMH)}\asymp
\frac{(2\pi)^{k/2}}{M^k_k}\left\{1-\Phi\left(\sqrt{\frac{k}{2}}\right)\right\}.
\label{eq:RWMH_asymp}
\end{equation}

\subsubsection{Asymptotic bounds of the acceptance rate for additive TMCMC}
\label{subsubsec:asymp_tmcmc}
Next let us obtain lower and upper bounds of $AR^{(TMCMC)}$ associated with TMCMC with additive transformation. 
Recall that in this case, $\mathbf{x}' = \mathbf{x} + \epsilon \boldsymbol\delta$ 
where $\epsilon \sim N(0,1)\mathbb{I}(\epsilon > 0)$ and the components $\delta_i$ 
of $\boldsymbol\delta$ are $iid$ taking values $\pm1$ with probability $1/2$ each. 
In this set up (\ref{eqn:new_bound2}) becomes 
\begin{equation}\label{eqn:new_bound3}
 \begin{split}
  & Pr~\left(\frac{(M_k-m_k)}{2}(\mathbf{x}-\mathbf{x^*})^T (\mathbf{x}-\mathbf{x^*}) 
  -m_k\e(\mathbf{x}-\mathbf{x^*})^T\bdelta - \frac{m_k}{2}k\e^2
  > \log r\right) \\
 & \geq Pr~(R(\mathbf{x}'| \mathbf{x}) > r) \\
 & \geq 
   Pr~\left(-\frac{(M_k-m_k)}{2}(\mathbf{x}-\mathbf{x^*})^T (\mathbf{x}-\mathbf{x^*}) 
  -M_k\e(\mathbf{x}-\mathbf{x^*})^T\bdelta - \frac{M_k}{2}k\e^2
  > \log r\right),
 \end{split}
\end{equation}
Now notice that, under the lower bound of $\pi(\bx)$ provided in (\ref{eq:eq2}), as $k\rightarrow\infty$, 
$$\frac{M_k(\mathbf{x}-\mathbf{x^*})^T (\mathbf{x}-\mathbf{x^*})}{k}
\stackrel{\a. s.}{\longrightarrow} 1,$$ and
$$\frac{\sqrt{M_k}(\mathbf{x}-\mathbf{x^*})^T\bdelta}{k} \stackrel{\a. s.}{\longrightarrow}0.$$
Similarly, under the upper bound of $\pi(\bx)$ in (\ref{eq:eq2}), the above hold with $M_k$ replaced with $m_k$.
From these it follow that the asymptotic forms of the 
lower and the upper bounds of (\ref{eqn:new_bound3}) are given by
\[
Pr~\left(-\frac{(M_k-m_k)}{2}(\mathbf{x}-\mathbf{x^*})^T (\mathbf{x}-\mathbf{x^*}) 
  -M_k\e(\mathbf{x}-\mathbf{x^*})^T\bdelta - \frac{M_k}{2}k\e^2
  > \log r\right)
\]
\[ \asymp
\frac{(2\pi)^{k/2}}{M^k_k}\pi(\bx^*) 
\left\{2\Phi\left(\sqrt{-\frac{2}{kM_k}\log r-\left(\frac{M_k-m_k}{M^2_k}\right)}\right)-1\right\}
\]
and
\[
Pr~\left(\frac{(M_k-m_k)}{2}(\mathbf{x}-\mathbf{x^*})^T (\mathbf{x}-\mathbf{x^*}) 
  -m_k\e(\mathbf{x}-\mathbf{x^*})^T\bdelta - \frac{m_k}{2}k\e^2
  > \log r\right)
\]
\[ \asymp
\frac{(2\pi)^{k/2}}{m^k_k}\pi(\bx^*)
\left\{2\Phi\left(\sqrt{-\frac{2}{km_k}\log r+\left(\frac{M_k-m_k}{m^2_k}\right)}\right)-1\right\}. 
\]
Using the above results, it follows as in the case of $AR^{(RWMH)}$ that
\begin{align} 
&\frac{(2\pi)^{k/2}}{M^k_k}\pi(\bx^*)
\left\{2\inf_{u\in (\psi_1,1-\psi_2)}\Phi\left(\sqrt{-\frac{2}{kM_k}\log u-\left(\frac{M_k-m_k}{M^2_k}\right)}\right)-1\right\}\notag\\
&\leq AR^{(TMCMC)}
\leq \frac{(2\pi)^{k/2}}{M^k_k}\pi(\bx^*) 
\left\{2\sup_{u\in (\psi_1,1-\psi_2)}\Phi\left(\sqrt{-\frac{2}{km_k}\log u-\left(\frac{m_k-M_k}{m^2_k}\right)}\right)-1\right\}.\notag
\end{align}
Substituting the infimum and supremum over $(\psi_1,1-\psi_2)$ we obtain 
\begin{align} 
&\frac{(2\pi)^{k/2}}{M^k_k}\pi(\bx^*)
\left\{2\Phi\left(\sqrt{-\frac{2}{kM_k}\log (1-\psi_2)-\left(\frac{M_k-m_k}{M^2_k}\right)}\right)-1\right\}\notag\\
&\leq AR^{(TMCMC)}
\leq \frac{(2\pi)^{k/2}}{m^k_k}\pi(\bx^*)  
\left\{2\Phi\left(\sqrt{-\frac{2}{km_k}\log \psi_1-\left(\frac{m_k-M_k}{m^2_k}\right)}\right)-1\right\}.\notag
\end{align}
Since $k\left\vert\frac{M_k}{m_k}-1\right\vert\rightarrow 0$ and $m_k\asymp M_k$, it follows that
\begin{align} 
& -\frac{2}{kM_k}\log (1-\psi_2)-\left(\frac{M_k-m_k}{M^2_k}\right)
\asymp 
-\frac{2}{kM_k}\log (1-\psi_2)\quad\mbox{and}\notag\\
& -\frac{2}{km_k}\log \psi_1-\left(\frac{m_k-M_k}{m^2_k}\right)
\asymp 
-\frac{2}{km_k}\log \psi_1\asymp -\frac{2}{kM_k}\log \psi_1.\notag
\end{align}
Hence,  
\begin{eqnarray}
\frac{(2\pi)^{k/2}}{M^k_k}\pi(\bx^*)
\left\{2\Phi\left(\sqrt{-\frac{2}{kM_k}\log (1-\psi_2)}\right)-1\right\}
&\leq & AR^{(TMCMC)}\nonumber\\
&\leq & \frac{(2\pi)^{k/2}}{M^k_k}\pi(\bx^*)  
\left\{2\Phi\left(\sqrt{-\frac{2}{kM_k}\log \psi_1}\right)-1\right\}.\nonumber\\
\label{eq:AR_TMCMC_bounds}
\end{eqnarray}

For comparing (\ref{eq:AR_TMCMC_bounds}) with (\ref{eq:RWMH_asymp}) 
where $M_k=O\left(k^t\right);t>2$, it can be easily verified using L'Hospital's rule that
for any $\zeta_1>0$, $\zeta_2>0$,
\begin{equation}
\frac{2\Phi\left(\frac{\zeta_1}{\sqrt{kM_k}}\right)-1}{1-\Phi\left(\zeta_2\sqrt{k}\right)}
\rightarrow \infty.
\label{eq:comparison1}
\end{equation}
The above result will continue to hold if instead of $M_k=O\left(k^t\right);t>2$, 
$M_k\rightarrow a$, where $a>0$ is some constant.
Hence, $AR^{(TMCMC)}$
converges to zero at a much slower rate compared to $AR^{(RWMH)}$.

\section{Comparison of TMCMC with HMC}
\label{sec:hmc}

Motivated by Hamiltonian dynamics, \ctn{Duane87} introduced HMC, an MCMC algorithm with 
 deterministic proposals based on approximations of the Hamiltonian equations. 
 We will show that this algorithm is a special case of TMCMC, but first we provide a brief
 overview of HMC. More details can be found in \ctn{Liu01}, \ctn{Cheung09} and the references therein.

\subsection{Overview of HMC}
\label{subsec:hmc_overview}

If $\pi(\bx)$ is the target distribution, a fictitious dynamical system may be considered, where
$\bx(t)\in\mathbb R^d$ can be thought of as the $d$-dimensional position vector of a body of particles
at time $t$. If $\bv(t)=\dot{\bx}(t)=\frac{d{\bx}}{dt}$ is the speed vector of the
particles, $\dot{\bv}(t)=\frac{d{\bv}}{dt}$ is its acceleration vector, and 
$\Vec{F}$ is the force exerted on the particle; then, by Newton's law of motion 
$\Vec{F}=\boldm\dot{\bv}(t)=(m_1\dot{v_1},\ldots,m_d\dot{v_d})(t)$, where $\boldm\in\mathbb R^d$
is a mass vector. The momentum vector, $\bp=\boldm\bv$, often used in classical mechanics,
can be thought of as a vector of auxiliary variables brought in to facilitate simulation from $\pi(\bx)$.
The kinetic energy of the system is defined as $W(\bp)=\bp'\bM^{-1}\bp$, $\bM$ being the mass matrix.
Usually, $\bM$ is taken as $\bM=diag\{m_1,\ldots,m_d\}$.

The target density $\pi(\bx)$ is linked to the dynamical system via the potential energy field
of the system, defined as $U(\bx)=-\log\pi(\bx)$. The total energy (Hamiltonian function), is given by
$H(\bx,\bp)=U(\bx)+W(\bp)$.
A joint distribution over the phase-space $(\bx,\bp)$ is then considered, given by
\begin{equation}
f(\bx,\bp)\propto\exp\left\{-H(\bx,\bp)\right\}=\pi(\bx)\exp\left(-\bp'\bM^{-1}\bp/2\right)
\label{eq:phase_space_dist}
\end{equation}
Since the marginal density of $f(\bx,\bp)$ is $\pi(\bx)$, it now remains to provide a joint proposal
mechanism for simulating $(\bx,\bp)$ jointly; ignoring $\bp$ yields $\bx$ marginally from $\pi(\cdot)$. 

For the joint proposal mechanism, HMC makes use of Newton's law of motion, derived from
the law of conservation of energy, and often written in the form of Hamiltonian equations, given by 
\begin{eqnarray}
\dot{\bx}(t)&=&\frac{\partial H(\bx,\bp)}{\partial\bp}=\bM^{-1}\bp,\nonumber\\
\dot{\bp}(t)&=&-\frac{\partial H(\bx,\bp)}{\partial\bx}=-\nabla U(\bx),\nonumber
\end{eqnarray}
where $\nabla U(\bx)=\frac{\partial U(\bx)}{\partial\bx}$.
The Hamiltonian equations can be approximated by the commonly used leap-frog algorithm (\ctn{Hockney70}), 
given by,
\begin{align}
\bx(t+\d t)&=\bx(t)+\d t\bM^{-1}\left\{\bp(t)-\frac{\d t}{2}\nabla U\left(\bx(t)\right)\right\}\label{eq:frog1}
\\[1ex]
\bp(t+\d t)&=\bp(t)-\frac{\d t}{2}\left\{\nabla U\left(\bx(t)\right)+\nabla U\left(\bx(t+\d t)\right)\right\}\label{eq:frog2}
\end{align}
Given choices of $\bM$, $\d t$, and $L$, the HMC is then given by the following algorthm:
\begin{algo}\label{algo:hmc} \topline HMC \botline \normalfont \ttfamily
\begin{itemize}
 \item Initialise $\bx$ and draw $\bp\sim N(\bzero, \bM)$.
 \item Assuming the current state to be $(\bx,\bp)$, do the following:
\begin{enumerate}
 \item Generate $\bp'\sim N\left(\bzero,\bM\right)$;
 \item Letting $(\bx(0),\bp(0))=(\bx,\bp')$, run the leap-frog algorithm for 
 $L$ time steps, to yield $(\bx'',\bp'')=\left(\bx(t+L\d t),\bp(t+L\d t)\right)$;
\item Accept $(\bx'',\bp'')$ with probability 
\begin{equation}
\min\left\{1,\exp\left\{-H(\bx'',\bp'')+H(\bx,\bp')\right\}\right\},
\label{eq:hmc_accept}
\end{equation}
and accept $(\bx,\bp')$ with the remaining probability.
\end{enumerate}
\end{itemize}
\botline \rmfamily
\end{algo}
In the above algorithm, it is not required to store simulations of $\bp$.
Next we show that HMC is a special case of TMCMC.

\subsection{HMC is a special case of TMCMC}
\label{subsec:hmcmc_special_tmcmc}

To see that HMC is a special case of TMCMC, note that the leap-frog step of the HMC algorithm (Algorithm \ref{algo:hmc})
is actually a deterministic transformation of the form $g^L: (\bx(0),\bp(0))\to(\bx(L),\bp(L))$ (see \ctn{Liu01}). 
This transformation satisfies the following: if $(\bx',\bp')=g^L(\bx,\bp)$, then $(\bx,-\bp)=g^L(\bx',-\bp')$.

The Jacobian
of this transformation is 1 because of the volume preservation property, which says that if
$V(0)$ is a subset of the phase space, and if $V(t)=\left\{(\bx(t),\bp(t)): (\bx(0),\bp(0))\in V(0)\right\}$, then the volume
$|V(t)|={\int\int}_{V(t)}d\bx d\bp={\int\int}_{V(0)}d\bx d\bp=|V(0)|$. As a result, the Jacobian does not feature in the
HMC acceptance probability (\ref{eq:hmc_accept}).

For any dimension, there is only one move type defined for HMC,
which is the forward transformation $g^L$. Hence, this move type has probability one of selection, and all
other move types which we defined in general terms in connection with TMCMC, have zero probability of selection. 
As a result, 
the corresponding TMCMC acceptance ratio needs 
slight modification---it must be made free of the move-type probabilities, 
which is exactly the case in (\ref{eq:hmc_accept}).

The momentum vector
$\bp$ can be likened to $\be$ of TMCMC, but note that $\bp$ must always be of the same dimensionality as $\bx$;
this is of course, permitted by TMCMC as a special case. 

\subsection{Comparison of acceptance rate for $L=1$ with RWMH and TMCMC}
\label{subsec:hmc_vs_tmcmc}

For $L=1$, the proposal corresponding to HMC is given by (see \ctn{Cheung09})
\begin{equation}
q(\bx'\mid\bx(t))=N\left(\bx':\bmu(t),\bSigma(t)\right),
\label{eq:hmc_proposal}
\end{equation}
where (\ref{eq:hmc_proposal}) is a normal distribution with mean and variance
given, respectively, by the following:
\begin{align}
\bmu(t)&=\bx(t)+\frac{1}{2}\bM^{-1}\d t\nabla\log\left(\pi(\bx(t))\right)\label{eq:hmc_mean}
\\[1ex]
\bSigma(t)&=\d t\bM^{-1}\label{eq:hmc_var}
\end{align}
Assuming diagonal $\bM$ with $m_i$ being the $i$-th diagonal element,
the proposal can be re-written in the following more convenient manner:
for $i=1,\ldots,k$,
\begin{equation}
x'_i=x_i(t)+\e_i,
\label{eq:hmc_proposal2}
\end{equation}
where $s_i(t)$ denotes the $i$-th component of $\nabla\log\left(\pi(\bx(t))\right)$,
and $\e_i\sim N\left(\frac{1}{2}\frac{\d t s_i(t)}{m_i},\frac{\d t}{m_i}\right)$.
Assuming, as is usual, that $m_i=1$ for each $i$, it follows that
\begin{equation}
\frac{\parallel\bx'-\bx\parallel^2}{\d t^2} =\sum_{i=1}^k\left(\frac{\e_i}{\d t}\right)^2=\sum_{i=1}^k\e'^2_i
\sim\chi^2_k(\lambda),
\label{eq:hmc_rate1}
\end{equation}
where $\chi^2_k(\lambda)$ is a non-central $\chi^2$ distribution with $k$ degrees of freedom
and non-centrality parameter 
$\lambda =\frac{\d t^2}{4}\sum_{i=1}^ks^2_i(t)$.
Since, as either $k\rightarrow\infty$ or $\lambda\rightarrow\infty$,
\begin{equation}
\frac{\chi^2_k(\lambda)-(k+\lambda)}{\sqrt{2(k+2\lambda)}}\stackrel{\mathcal L}{\rightarrow} N(0,1),
\label{eq:clt}
\end{equation}
assuming the same strong log-concavity conditions on the target density $\pi$ 
as provided in Section \ref{subsec:lower_upper_bound} it follows as in (\ref{eq:AR_RWMH_bounds}) that,
\begin{align}\label{eq:AR_HMC_bounds}
\frac{(2\pi)^{k/2}}{M^k_k}\pi(\mathbf x^*)
&\left\{1-\Phi\left(\frac{\log (1-\psi_2)+\frac{k}{2}\left[\left(\frac{M_k-m_k}{M_k}\right)+M_k\d t^2\left(1+\frac{\lambda}{k}\right)\right]}
{\sqrt{\frac{k}{2}\left[\left(\frac{M_k-m_k}{M_k}\right)^2+2M_k\d t^2\left(1+\frac{\lambda}{k}\right)
+M^2_k\d t^4\left(1+\frac{2\lambda}{k}\right)\right]}}\right)
\right\}
\leq  AR^{(HMC)}\notag\\
&\leq \frac{(2\pi)^{k/2}}{m^k_k}\pi(\mathbf x^*)
\left\{1-\Phi\left(\frac{\log \psi_1-\frac{k}{2}\left[\left(\frac{M_k-m_k}{m_k}\right)-m_k\d t^2\left(1+\frac{\lambda}{k}\right)\right]}
{\sqrt{\frac{k}{2}\left[\left(\frac{M_k-m_k}{m_k}\right)^2+2m_k\d t^2\left(1+\frac{\lambda}{k}\right)
+m^2_k\d t^4\left(1+\frac{2\lambda}{k}\right)\right]}}\right)\right\}\nonumber\\
\end{align}
If $\lambda/k\rightarrow 0$ as $k\rightarrow\infty$, it follows as in Section \ref{subsubsec:asymp_rwmh} that
\begin{equation}
AR^{(HMC)}\asymp
\frac{(2\pi)^{k/2}}{M^k_k}\left\{1-\Phi\left(\sqrt{\frac{k}{2}}\right)\right\},
\label{eq:HMC_asymp}
\end{equation}
which is of the same asymptotic form as (\ref{eq:RWMH_asymp}), corresponding to the RWMH acceptance rate.
On the other hand, if $\lambda/k\rightarrow \infty$ as $k\rightarrow\infty$, then it follows that
\begin{equation}
AR^{(HMC)}\asymp
\frac{(2\pi)^{k/2}}{M^k_k}\left\{1-\Phi\left(\frac{\sqrt{\frac{k}{2}\left(1+\frac{\lambda}{k}\right)}}
{\sqrt{2}\sqrt{\frac{1}{M_k\d t^2}+1}}
\right)\right\},
\label{eq:HMC_asymp2}
\end{equation}
which clearly tends to zero at a much faster rate compared to (\ref{eq:HMC_asymp}).

To summarize, if $\lambda/k\rightarrow 0$ as $k\rightarrow\infty$, then both HMC and RWMH have the
same asymptotic acceptance rate, tending to zero much faster than that of additive TMCMC. On the other hand,
if $\lambda/k\rightarrow\infty$ as $k\rightarrow\infty$, the acceptance rate of HMC tends to zero much faster
than that of RWMH, while that of additive TMCMC maintains its slowest convergence rate to zero. 
Also observe that the above conclusions will continue to hold if $m_k$ and $M_k$ tend to finite positive
constants satisfying $M_k\asymp m_k$ and $k\left\vert\frac{m_k}{M_k}-1\right\vert\rightarrow 0$ as 
$k\rightarrow\infty$.


\section{Generalized Gibbs/Metropolis approaches and comparisons with TMCMC}
\label{sec:liu}


It is important to make it clear at the outset of this discussion that the goals of 
TMCMC and generalized Gibbs/Metropolis
methods are different, even though both use moves based on transformations. 
While the strength of the latter lies in improving mixing of the standard Gibbs/MH algorithms
by adding transformation-based steps to the underlying collection of usual Gibbs/MH steps, 
TMCMC is an altogether general method
of simulating from the target distribution which does not require any underlying step of Gibbs or
MH.  

The generalized Gibbs/MH methods work in the following manner.
Suppose that an underlying Gibbs or MH algorithm for exploring a target distribution
has poor mixing properties. Then in order to improve mixing, one may consider some suitable transformation
of the random variables being updated such that mixing is improved under the transformation.
Such a transformation needs to chosen carefully since it is important to ensure that 
invariance of the Markov chain is preserved under the transformation.
It is convenient to begin with an overview of the generalized Gibbs method
with a sequential updating scheme and then proceed to the discussion on the 
issues and the importance of the block updating idea in the context of improving
mixing of standard Gibbs/MH methods.

\ctn{Liu00} (see also \ctn{Liu99}) propose simulation of a transformation 
from some appropriate probability distribution, and then applying the transformation
to the co-ordinate to be updated. For example, in a $d$-dimensional
target distribution, for updating $\bx=(x_1,x_2,\ldots,x_d)$
to $\bx'=(x'_1,x_2,\ldots,x_d)$, using an additive transformation, 
one can select $\e$ from some appropriate distribution
and set $x'_1=x_1+\e$. Similarly, if a scale transformation is desired, then 
one can set $x'_1=\gamma x_1$, where $\gamma$
must be sampled from some suitable distribution. The suitable distributions of $\e$ and
$\gamma$ are chosen such that the target distribution is invariant with respect to
the move $\bx'$, the forms of which are provided in \ctn{Liu00}. 
For instance, if $\pi(\cdot)$ denotes the target distribution, then for the additive
transformation, $\e$ may be sampled from $\pi(x_1+\e,x_2,\ldots,x_d)$, and
for the multiplicative transformation, one may sample $\gamma$ 
from $|\gamma|\pi(\gamma x_1,x_2,\ldots,x_d)$.
Since direct sampling from such distributions may be impossible, \ctn{Liu00} suggest
a Metropolis-type move with respect to a transformation-invariant transition kernel.

Thus, in the generalized Gibbs method, sequentially all the variables must be updated,
unlike TMCMC, where all the variables can be updated simultaneously in a single block.
Here we note that for irreducibility issues the generalized Gibbs approach is not suitable for
updating the variables blockwise using some transformation that acts on all the variables in a given block. 
To consider a simple example, with say, $d=2$ and a single block consisting of both the variables, if
one considers the additive transformation, then starting with $\bx=(x_1,x_2)$, where $x_1<x_2$,
one can not ever reach $\bx'=(x'_1,x'_2)$, where $x'_1>0, x'_2<0$. This is because $x'_1=x_1+z$
and $x'_2=x_2+z$, for some $z$, and $x'_1>0,x'_2<0$ implies $z>-x_1$ and $z<-x_2$, which is a contradiction.
The scale transformation implies the move 
$\bx=(x_1,\ldots,x_d)\rightarrow (\gamma x_1,\ldots,\gamma x_d)=\bx'$. If
one initializes the Markov chain with all components positive,
for instance, then in every iteration, all the variables will have the same sign.
The spaces where some variables are positive and some negative will never be visited,
even if those spaces have positive (in fact, high) probabilities under the target distribution.
This shows that the Markov chain is not irreducible.
In fact, with the aforementioned approach, no transformation,
whatever distribution they are generated from, can guarantee irreducibility in general if 
blockwise updates using the transformation strategy of generalized Gibbs is used. 

Although blockwise transformations are proposed in \ctn{Liu00} (see also \ctn{Kou05} who propose
a MH-based rule for blockwise transformation), they are meant for a different purpose than that
discussed above. The strength of such blockwise transformations lies in improving the mixing behaviour
of standard Gibbs or MH algorithms.
Suppose that an underlying Gibbs or MH algorithm for exploring a target distribution
has poor mixing properties. Then in order to improve mixing, one may consider some suitable transformation
of the set of random variables being updated such that mixing is improved under the transformation.
This additional step involving transformation of the block of random variables can be obtained
by selecting a transformation from the appropriate probability distribution provided in \ctn{Liu00}.
This ``appropriate" probability distribution guarantees that stationarity of the transformed block 
of random variables is preserved. Examples reported in \ctn{Liu00}, \ctn{Muller05}, \ctn{Kou05}, etc. 
demonstrate that this 
transformation also improves the mixing behaviour of the
chain, as desired. 

Thus, to improve mixing using the methods of \ctn{Liu00} or \ctn{Kou05} one needs to run the usual Gibbs/MH 
steps, with an additional step involving transformations as discussed above. This additional step
induces more computational burden compared to the standard Gibbs/MH steps, but improved mixing
may compensate for the extra computational labour. In very high dimensions, of course, this
need not be a convenient approach since computational complexity usually makes standard Gibbs/MH approaches
infeasible. Since the additional transformation-based step works on the samples generated by
standard Gibbs/MH, impracticality of the latter implies that the extra transformation-based step of \ctn{Liu00}
for improving mixing is of little value in such cases.

It is important to point out that the generalized Gibbs/MH methods can be usefully employed by
even TMCMC to further improve its mixing properties. In other words, a step of generalized
Gibbs/MH can be added to the computational fast TMCMC. This additional step can significantly
improve the mixing properties of TMCMC. That TMCMC is much faster computationally than
standard Gibbs/MH methods imply that even in very high-dimensional situations the generalized
Gibbs/MH step can ve very much successful while working in conjunction with TMCMC.


\begin{table}[ht]
\begin{small}
 \begin{center}
\begin{tabular}{c|c|c||c|c|c} \hline
Flight no. &  Failure & Temp & Flight no. &  Failure & Temp\\ \hline
 14 & 1 & 53 & 2  & 1 &  70 \\
9  & 1 &  57 & 11  & 1 &  70 \\
23 &  1  & 58 & 6  & 0 &  72 \\
10 &  1 &  63  & 7  & 0 &  73 \\
1 &  0 &  66 & 16  & 0 &  75 \\
5 &  0 &  67 & 21  & 1 &  75 \\
13 &  0 &  67 & 19  & 0 &  76 \\
15  & 0 &  67 & 22  & 0 &  76 \\
4  & 0 &  68 & 12  & 0 &  78 \\
3  & 0 &  69 & 20  & 0 &  79 \\
8  & 0 &  70 & 18  & 0 &  81 \\
17  & 0 &  70 & & & \\ \hline
\end{tabular}
\end{center}
\caption{Challenger data. Temperature at flight time (degrees F) and failure of 
O-rings (1 stands for failure, 0 for success).}
\label{tab:challenger}
\end{small}
\end{table}

\section{Examples of TMCMC for discrete state spaces}
\label{sec:discrete_tmcmc}
The ideas developed in this paper are not confined to continuous target distributions, but also
to discrete cases. For the sake of illustration, we consider two examples below.
\begin{itemize}
\item[(i)] Consider an Ising model, where, for 
$i=1,\ldots,k$ $(k\geq 1)$, the discrete random variable $x_i$ takes the value $+1$ or $-1$
with positive probabilities. We then have
$\statesp = \{-1,1\}$. To implement TMCMC, consider the 
forward transformation 
$T(x_i,\e) = sgn(x_i+\e)$ with probability $p_i$, and choose the backward transformation as  
$T^b(x_i,\e) = sgn(x_i-\e)$ with probability $1-p_i$. Here $sgn(a)=\pm 1$ accordingly as $a>0$ or $a<0$, and $\Y = (1,\infty)$.
Note the difference with the continuous cases. Here even though neither of the transformations is 1-to-1
or onto, TMCMC works because of discreteness; the algorithm can easily be seen to satisfy detailed balance,
irreducibility and aperiodicity.
However, if $k=1$ with $x_1$ being the only variable, then, if $x_1=1$, it is possible to choose, 
with probability one, the backward move-type, yielding
$T^b(x_1,\e)=-1$. On the other hand, if $x_1=-1$, with probability one, we can choose the forward
move-type, yielding $T(x_1,\e)=1$. 
Only $2^k$ move-types are necessary
for the $k$-dimensional case for one-step irreducibility.
In discrete cases,
however, there will be no Jacobian of transformation, thereby simplifying the acceptance ratio.

\item[(ii)] For discrete state spaces like $\mathbb Z^k$, ($\mathbb Z = \{0,\pm1,\pm2,\ldots\}$) the additive transformation with single epsilon does not work. For example, with $ k =2$, if the starting state is $(1,2)$ then the chain will never reach any states $(x,y)$ where $x$ and $y$ have same parity (i.e. both even or both odd) resulting a reducible Markov chain. Thus in this case we need to have more move-types than $2^k$. For example, with some positive probability (say $r$) we may select a random coordinate and update it leaving other states unchanged. With the remaining probability (i.e. $1-r$) we may do the analogous version of the additive transformation:\\
 Let $\Y=[1,\infty)$. Then, can choose the forward transformation for each coordinate as
$T_i(x_i,\e)=x_i+[\e]$ and the backward transformation as $T^b_i(x_i,\e)=x_i-[\e]$, where $[a]$ denotes
the largest integer not exceeding $a$.

This chain is clearly ergodic and we still need only \emph{one} epsilon to update the states. 

\end{itemize}
However, in discrete cases, TMCMC reduces to Metropolis-Hastings with a mixture proposal.
But it is important to note that the implementation is much efficient and computationally cheap when TMCMC-based
methodologies developed in this paper, are used.